%% file: main.tex
\documentclass[conference]{IEEEtran}
\IEEEoverridecommandlockouts
\usepackage{amsmath,amssymb,amsfonts}
\usepackage{graphicx}
\usepackage{textcomp}
\usepackage{xcolor}

\def\BibTeX{{\rm B\kern-.05em{\sc i\kern-.025em b}\kern-.08em
    T\kern-.1667em\lower.7ex\hbox{E}\kern-.125emX}}

\usepackage{booktabs} % for professional tables
\usepackage{url}            % simple URL typesetting
\usepackage{amsfonts}       % blackboard math symbols
\usepackage{nicefrac}       % compact symbols for 1/2, etc.
\usepackage{microtype}      % microtypography
\usepackage{graphicx}
\usepackage{doi}
\usepackage{multicol}
\usepackage{multirow}
\usepackage{bm, dsfont}
\usepackage{mathrsfs}
\usepackage{hyperref}
\usepackage{caption}
\usepackage{algorithm,algpseudocode}
\usepackage[T1]{fontenc}
\usepackage[natbib=true, sorting=none]{biblatex}
\usepackage{mathtools}
\usepackage{amsthm}

\usepackage[capitalize,noabbrev]{cleveref}

\usepackage[normalem]{ulem}

\theoremstyle{plain}
\newtheorem{theorem}{Theorem}[section]

\newtheorem{corollary}[theorem]{Corollary}
\theoremstyle{definition}
\newtheorem{definition}[theorem]{Definition}

\theoremstyle{remark}

\DeclareMathOperator*{\argmax}{arg\,max}

\DeclareMathOperator{\argmaxinline}{arg\,max}
\DeclareMathOperator{\argmininline}{arg\,min}
\newcommand{\matr}[1]{\mathbf{#1}}
\newcommand{\vect}[1]{\mathbf{#1}}
\newcommand{\vectsym}[1]{\bm{#1}}
\newcommand{\set}[1]{\mathcal{#1}}
\DeclarePairedDelimiter\abs{\lvert}{\rvert}%
\DeclarePairedDelimiter\norm{\lVert}{\rVert}%
\usepackage[disable,textsize=tiny]{todonotes}
\usepackage{fancyhdr}

\fancypagestyle{specialfooter}{%
  \fancyhf{}
  
  \fancyfoot[R]{ \noindent\fbox{%
    \parbox{\textwidth}{%
        {\footnotesize \copyright 20xx IEEE. Personal use of this material is permitted. Permission from IEEE must be obtained for all other uses, in any current or future media, including reprinting/republishing this material for advertising or promotional purposes, creating new collective works, for resale or redistribution to servers or lists, or reuse of any copyrighted component of this work in other works.}
        }
    }}
}

\algnewcommand{\IIf}[1]{\State\algorithmicif\ #1\ \algorithmicthen}
\algnewcommand{\EndIIf}{\unskip\ \algorithmicend\ \algorithmicif}
\begin{document}

\newcommand{\gao}[1]{\todo{NG: #1}}
\newcommand{\tom}[1]{\todo{TW: #1}}
\newcommand{\nic}[1]{\todo{NF: #1}}

\title{Quantum Robustness Verification: A Hybrid Quantum-Classical Neural Network Certification Algorithm
\thanks{The project/research is supported by the Bavarian Ministry of Economic Affairs, Regional Development and Energy with funds from the Hightech Agenda Bayern.}
}

\author{
    \IEEEauthorblockN{Nicola Franco\IEEEauthorrefmark{2}, Tom Wollschl\"ager\IEEEauthorrefmark{3}, Nicholas Gao\IEEEauthorrefmark{3}, Jeanette Miriam Lorenz\IEEEauthorrefmark{2}, Stephan G\"unnemann\IEEEauthorrefmark{3}}
    \IEEEauthorblockA{\IEEEauthorrefmark{2}Fraunhofer Institute for Cognitive Systems IKS, Munich, Germany
    \\\{nicola.franco, jeanette.miriam.lorenz\}@iks.fraunhofer.de}
    \IEEEauthorblockA{\IEEEauthorrefmark{3}Dept. of Informatics \& Munich Data Science Institute, Technical Univ. of Munich, Germany
    \\\{tom.wollschlaeger, gaoni, guennemann\}@in.tum.de}
}

\maketitle
\thispagestyle{specialfooter} % Only for the pre-print

\begin{abstract}
In recent years, quantum computers and algorithms have made significant progress indicating the prospective importance of quantum computing (QC). 
Especially combinatorial optimization has gained a lot of attention as an application field for near-term quantum computers, both by using gate-based QC via the Quantum Approximate Optimization Algorithm and by quantum annealing using the Ising model. 
However, demonstrating an advantage over classical methods in real-world applications remains an active area of research. 
In this work, we investigate the robustness verification of ReLU networks, which involves solving many-variable mixed-integer programs (MIPs), as a practical application. 
Classically, complete verification techniques struggle with large networks as the combinatorial space grows exponentially, implying that realistic networks are difficult to be verified by classical methods. 
To alleviate this issue, we propose to use QC for neural network verification and introduce a hybrid quantum procedure to compute provable certificates. By applying Benders decomposition, we split the MIP into a quadratic unconstrained binary optimization and a linear program which are solved by quantum and classical computers, respectively.
We further improve existing hybrid methods based on the Benders decomposition by reducing the overall number of iterations and placing a limit on the maximum number of qubits required.
We show that, in a simulated environment, our certificate is sound, and provides bounds on the minimum number of qubits necessary to approximate the problem. Finally, we evaluate our method within simulations and on quantum hardware.
\end{abstract}


\input{sections/intro/introduction}
\input{sections/intro/related_work}
\input{sections/intro/background}

\input{sections/method/overview}
\input{sections/method/benders}
\input{sections/method/improvements}
\input{sections/method/bridge}
\input{sections/results}

\input{sections/conclusion}

% \section*{Acknowledgment}

\printbibliography

\end{document}

%% file: sections/intro/introduction.tex
\section{Introduction}

Despite recent breakthroughs of machine learning models, they have been shown to be brittle to \textit{adversarial examples} -- small perturbations added to an input designed by an adversary to deceive the model \cite{szegedy2014intriguing, goodfellow2014explaining}. Such adversarial examples differ only by an $\epsilon$ to the original input measured by some norm, e.g., in the image domain these perturbations are often imperceptible to the human eye. This vulnerability raises questions about the applicability of machine learning models in high-risk areas such as autonomous driving \cite{gnanasambandam2021optical} or fraud detection \cite{cartella2021adversarial}.

Consequently, research on (adversarial) robustness of machine learning models has gained importance. Even though empirical defenses, like specific training techniques \cite{kurakin2017adversarial} or architectures \cite{dhillon2018stochastic} improve the robustness of the neural networks and provide little inference time overhead, they are easy to break \cite{carlini2017adversarial, athalye2018obfuscated, goodfellow2018gradient}. Provable certificates for robustness are advantageous as they guarantee the robustness of an instance to any adversarial attack with an admissible perturbation set, e.g., in our case being the infinity norm ball around the input. However, they come at a large computational cost as the search space grows exponentially with increasing network sizes. Even for ReLU networks, which are especially attractive due to the simple nature of their activation function, the complete verification problem has been shown to be NP-complete \cite{katz2017reluplex}.

%, motivating the research towards approximate solutions. These methods approximate exact approaches to decrease computational complexity \cite{DBLP:journals/corr/abs-2007-10868, wang2021betacrown}. However, they introduce a trade-off between tightness of the certificates and runtime, e.g., within the convex approximation of ReLU networks \cite{wong2018provable} one may not know whether an input is robust against attacks.
%Here, we take another approach to verification via quantum computing-assisted methods.

%The development of quantum computing (QC) hardware and software has seen significant progress over the recent years, e.g., recently quantum supremacy has been shown in an academic sample~\cite{Arute:2019zxq}. However, practical applications of QC remains an active area of research. %
%Current quantum computers belong to the category of Noisy Intermediate Scale Quantum (NISQ) devices, as they feature a relatively small number (around 100) qubits, are noisy and show a relatively small connectivity between the qubits. Therefore, it is infeasible to perform large and complex calculation solely on a quantum computer. 
%Instead, current algorithms are of hybrid nature and require an interaction between classical or high performance computers and the quantum computer.\gao{this paragraph is to broad and unspecific, see my suggestion below.}

The development of quantum computing (QC) hardware and software has seen significant progress over the recent years, e.g., recently quantum supremacy has been shown in an academic sample~\cite{Arute:2019zxq}. However, practical applications of QC remain an active area of research as current quantum computers belong to the category of Noisy Intermediate Scale Quantum (NISQ) devices.
They are characterized by their relatively small number (around 100) qubits, noise and relatively small connectivity between the qubits. Therefore, it is infeasible to perform large and complex calculation solely on a quantum computer. %
Instead, current algorithms are of hybrid nature and require an interaction between classical and quantum computers.
With the motivation of exploiting the superposition of the exponentially large space of qubits, entanglement and interference effects,  combinatorial optimization has received a lot of attention both within the gate-based QC and quantum annealing (QA) via the Quantum Approximate Optimization Algorithm (QAOA) \cite{farhiQuantumApproximateOptimization2014} and the Ising model, respectively~\cite{dasQuantumAnnealingRelated2005}.

In this work, we propose to leverage the combinatorial power of QC to verify the robustness of ReLU networks by proposing a hybrid quantum-classical algorithm. 
Since the ReLU non-linearity may be reformulated as a binary variable, one can represent the verification problem as a mixed-integer program (MIP) \cite{ehlers2017formal, tjeng2018evaluating}.
We apply Benders decomposition to partition the MIP into a quadratic unconstrained binary optimization (QUBO) and a linear program (LP) \cite{benders1962, chang2020quantum}. Thus, the resulting algorithm is an iterative process where the LP generates cuts for the QUBO which is solved by a quantum computer or a quantum annealer.
To the best of our knowledge, this is the first work of certifying neural networks with quantum computers. As the current hardware cannot compete with classical solvers, this algorithm should be understood as a proof of concept.
However, as the availability of increasingly more powerful computers increases, we may see polynomial speed-ups~\cite{groverFastQuantumMechanical1996}.

%We demonstrate how to use QC for the approach of exact \tom{maybe more careful here with exact} verification for ReLU networks through a mixed-integer program (MIP) \cite{ehlers2017formal, tjeng2018evaluating}. As proven by \citet{katz2017reluplex} the exact verification of neural networks is NP-complete. Hence, for large networks, QC-assisted solutions might provide up to\gao{Should we write up to?} a polynomial speed-up in comparison to classical solutions with the availability of more powerful QC-hardware\gao{cite grover's search}. Following the MIP formulation of the certificate, we leverage Benders decomposition to partition the optimization into both a purely binary program and a linear program (LP) \cite{benders1962, chang2020quantum}. Thus, the quadratic unconstrained binary optimization (QUBO) program can be solved by a quantum computer or a quantum annealing device with the resulting algorithm being an iterative process of the linear program generating cuts for the QUBO. 

Moreover, we show that our method is sound in an ideal, simulated setting and converges to the exact solution in the limit. For the modifications required for current QC hardware, we show that our optimization objective is conservative, i.e., the objective value of our approach is a lower bound to the one of the exact verification. Additionally, we derive bounds for the minimum number of qubits necessary and provide an experimental evaluation of our approach on a conventional computer, a quantum annealing simulation as well as real QC hardware. 

%In summary, our core contributions are to (i) introduce the Hybrid Quantum-Classical Robustness Approach for Neural networks \textit{HQ-CRAN} -- a procedure to use quantum computing power for neural network robustness verification, (ii) proof the soundness of the certificate under ideal conditions, (iii) provide bounds to quantify the deviation from the ideal scenario, (iv) an estimate for the number of qubits necessary to achieve reasonable certificate quality, and (v) contribute experiments on quantum annealers and quantum computers.\gao{if we have the space having this as list seems better}

In summary, our core contributions are:
\begin{enumerate}
    \item Introducing the first Hybrid Quantum-Classical Robustness Algorithm for Neural network verification \textit{HQ-CRAN}. %\tom{introducing the first quantum algorithm for neural network verification?}%-- a  hybrid quantum computing procedure for ReLU network certification
    \item Proving the soundness of the certificate.
    \item Estimating the number of qubits necessary to achieve reasonable certificate quality, absent in \citet{chang2020quantum}. %\tom{compared to the work of Chen et al? Should state here something}
    \item Decreasing the overall number of qubits and required amount of iterations to converge compared to \citet{chang2020quantum} and \citet{zhao2021hybrid}.
    \item Evaluating HQ-CRAN on quantum annealers and computers.
\end{enumerate}% 

%% file: sections/intro/related_work.tex
\section{Related Work}

\textbf{Adversarial Robustness }
The research field can be divided into two main categories, namely \textit{empirical-} and \textit{provable defenses}. While empirical defenses such as using specific network architectures \cite{Gong2017a, dhillon2018stochastic, Sitawarin2019}, adversarial training \cite{kurakin2017adversarial, Shafahi2019} or data preparation \cite{guo2018countering, Buckman2018} are easier to use in practice and require low additional inference time complexity, they are only evaluated against the current state-of-the-art adversarial attacks. However, it is not guaranteed that a newly introduced attack will not circumvent the defense \cite{athalye2018obfuscated, goodfellow2018gradient}. Contrary, a provable or certifiable verification means that the network is robust for a given perturbation budget regardless of the attack method. It follows that for all perturbed inputs within the budget, the prediction of the network remains the same.  

Within this line of research, \textit{exact} or \textit{complete} verification methods focus on reasoning about the adversarial polytope, i.e. the shape containing all possible outputs given the set of perturbed inputs. Since these exact verification methods \cite{bunel2020branch, tjeng2018evaluating, katz2017reluplex, ruan2018reachability} do not scale to large networks due to their exponential time complexity, \textit{incomplete} methods \cite{8418593, wong2018provable, mueller2021scaling, wang2021betacrown, xu2020automatic, dathathri2020enabling, muller2021prima} adopt convex approximations to overcome the non-linearity of the network and reduce the overall complexity. The resulting output of these approaches is the worst-case point of the exact adversarial polytope in case of a \textit{complete} method, and of an outer approximation in case of an \textit{incomplete} one. %\tom{it should be the "worst-case" point of the outer approximation. Similarly in the MIP it should be the worst-case of the actual polytope. But all in all it means that the program reasons over the adv. polytope or its outer approximation} 
Due to the outer approximation, the resulting certificates are not sharp, i.e. there might be cases where it fails to verify robustness even though the prediction of the sample could not be changed given the tested budget.

\textbf{Hybrid Decomposition Algorithms }
Recently, a multitude of quantum optimization algorithms have been proposed for NISQs. The common goal is to solve large-scale optimization problems that are otherwise computationally intractable by pure classical hardware. \citet{gambella2020multiblock} presented a decomposition method based on the alternating direction method of multipliers (ADMM). The method is a heuristic algorithm which splits a mixed-binary optimization problem into a QUBO, solved with QC, and a LP via a multi-block version of ADMM. Similarly, \citet{chang2020quantum} and \citet{zhao2021hybrid} used Benders decomposition to divide a MIP into a QUBO and a LP. To formulate the QUBO, the methods require the approximation of real variables into binaries (or qubits). In addition to this, at each iteration a new real variable is constantly added, leading to an always increasing number of qubits. In this work, we improve \cite{chang2020quantum, zhao2021hybrid} by reducing the overall number of iterations and by placing a limit to the maximum number of qubits required.

%% file: sections/intro/background.tex
\section{Background and Preliminaries} \label{sec:background}

\textbf{Notation } We use lower case Latin and Greek letters $a,b,\dots,\alpha,\beta,\dots$ for scalars, bold $\vect{a}$ for vectors, capitalized bold $\matr{A}$ for matrices, and calligraphic $\mathcal{A}$ for sets. Furthermore, we use $\vect{a}_i$ to denote the $i$-th element of $\vect{a}$ and $\lvert \cdot \rvert$ to denote the cardinality of a set. \texttt{diag}$\left(\vect{v}\right)$ denotes a diagonal matrix with all vector entries on the main diagonal and zero elsewhere. We denote with $\odot$ the element-wise product and with $\otimes$ the outer product between two vectors. The identity matrix in $\mathbb{R}^{n\times n}$ is $\mathbf{I}_n$, while $\mathbf{O}_{n \times m}$ is the zero matrix in $\mathbb{R}^{n \times m}$. The column vector of $n$ ones is denoted by $\mathds{1}_n$, and $\vect{0}_n$ is the zero column vector in $\mathbb{R}^{n}$.

\textbf{Model Formulation }
We define a neural network by a function $\vect{f}(\vect{x}) \colon \mathcal{X} \to \mathbb{R}^{\lvert \mathcal{Z} \rvert}$ which maps input samples $\vect{x} \in \mathcal{X}$ to output $\vect{z} \in \mathbb{R}^{\lvert \mathcal{Z} \rvert}$, where $\mathcal{Z}$ is the set of classes. We assume a feedforward architecture composed by affine transformations and followed by ReLU activation functions:
\begin{equation*}
\begin{aligned}
\hat{\vect{z}}^{[i]} &= \vect{W}^{[i]}\vect{z}^{[i-1]} + \vect{v}^{[i]}, \\
\vect{z}^{[i]} &= \max{\{0, \hat{\vect{z}}^{[i]}\}}, \quad \forall i \in \left\{1,\dots, L\right\},
\end{aligned}
\end{equation*}
where $L$ represents the number of layers, $\vect{z}^{[0]} \equiv \vect{x}$ and $\vect{f}(\vect{x}) \equiv \vect{z}^{[L]}$. In case of classification, the network outputs a vector in $\mathbb{R}^{\abs{\set Z}}$. The predicted class is then given by the index of the largest value of that vector, i.e. $c = \argmaxinline_{j}\vect{f}(\vect{x})_j$.

\textbf{Robustness Certificate and Threat Model } 
Let $\vect{x}$ be an input, e.g. a vectorized image. We say that a neural network $f$ is certifiably robust for this input if the prediction for all perturbed versions remains unchanged:
\begin{equation} \label{eq:certificate}
\argmaxinline_j \vect f(\vect{x})_j = \argmaxinline_j \vect f(\Tilde{\vect{x}})_j \quad \forall \Tilde{\vect{x}} \in \set{B}_\epsilon^p(\vect{x}).
\end{equation}
Here, $\epsilon$ is the perturbation budget and $\Tilde{\vect{x}}$ is an element from the perturbation set  $\set{B}_\epsilon^\infty(\vect{x})$ based on the infinity norm:
\begin{equation}
\set{B}_\epsilon^\infty(\vect{x}) = \{\Tilde{\vect{x}} \,|\,\norm{\vect{x} - \Tilde{\vect{x}}}_\infty \leq \epsilon\}. 
\end{equation}
If we cannot certify an input, it means that there exists $\vect{x}' \in \set{B}_\epsilon^\infty(\vect{x})$ for which $\argmaxinline_j\vect{f}(\vect{x})_j \neq \argmaxinline_j\vect{f(\vect{x}')}_j$. We call any of these $\vect{x}^\prime$ an adversarial example.

\textbf{Verification as Optimization } We can formulate \autoref{eq:certificate} as an optimization task. The robustness verification can be written as\footnote{This formulation is for a network that uses a ReLU activation after the last transformation layer, i.e. the output of the network is positive.}:
\begin{subequations} \label{eq:general_verification}
\begin{alignat}{2}
    &\min_{t} \min_{\vect{z}^{[i]} \in \mathbb{R}^{n_{i}}}&\qquad&\vect{z}^{[L]}_c - \vect{z}^{[L]}_t \label{eq:general_objective}\\
    &\,\text{subject to} &&\vect{z}^{[0]}\in \set{B}^{p}_{\epsilon}(\vect{x}), \\
    &&&\hat{\vect{z}}^{[i]} = \vect{W}^{[i]}\vect{z}^{[i-1]} + \vect{v}^{[i]}, \\
    &&&\vect{z}^{[i]} = \max\{0, \hat{\vect{z}}^{[i]}\}, \label{eq:general_relu}\\
    &&&\forall i \in \left\{1,\dots,L\right\}.
\end{alignat}
\end{subequations}
where $c$ is the index for the predicted class, $t \neq c$ is the target class and $\left\{1,\dots,L\right\}$ is the set of network layers. Thus, we compare the outputs of the neural network to predict any class other than the initially predicted one. If \autoref{eq:general_objective} is positive, it means that we are not able to change the prediction, i.e. not able to fool the network and can issue a robustness certificate. 

The piece-wise linear nature of ReLU activation units formulated in \autoref{eq:general_relu} characterizes problem~\autoref{eq:general_verification} as non-convex. There are two ways to solve this problem. Either we model the ReLU activation with an integer variable or we enclose the possible activation values $\vect{z}^{[i]}$ with a convex area. The former approach yields \textit{complete} formulation of the exact polytope but the integer variables render the problem NP-hard. The latter approach is an instance of an \textit{incomplete} solution.
%\tom{maybe we have to be more careful here that the approach only reasons about the exact polytope and is thus complete. but this can also be done in the camera ready version}

\textbf{Complete Verification }
In case of complete verification, we follow~\citet{tjeng2018evaluating} and formulate the ReLU activation units as a set of linear and integer constraints. For this formulation we assume that we have lower and upper bounds on the output of the network layer $\bm{\ell}, \vect u$ respectively. With the help of binary variables $\vect{y}$, we can then express the ReLU with linear constraints as: 
\begin{equation}\label{eq:complete_verification}
    \begin{aligned}
    &\vect{z} = \max\left\{0, \hat{\vect{z}}\right\}, \\
    &\bm{\ell} \leq \hat{\vect{z}} \leq \vect{u},
    \end{aligned}
    \quad \Longleftrightarrow \quad
    \begin{aligned}
        &\vect{z} \leq \hat{\vect{z}} - \bm{\ell} \odot \left( \mathds{1}_{n} - \vect{y} \right), \\ 
        &\vect{z} \geq \hat{\vect{z}}, \\
        &\vect{z} \leq \vect{u} \odot \vect{y}, \\
        &\vect{z} \geq \vect{0}_{n}, \\
        &\vect{y} \in \left\{0, 1\right\}^{n}.
    \end{aligned}
\end{equation}
Thus, to reason about the exact adversarial polytope we solve the program of \autoref{eq:general_verification} but replace the ReLU constraint \autoref{eq:general_relu} with \autoref{eq:complete_verification}. Note, the solution of the program is exact, independent of the tightness of the bounds $\bm{\ell}$ and $\vect{u}$. 

\textbf{Incomplete Verification }
There exist several approaches to decrease the runtime of the verification by sacrificing exactness. A straightforward way is to relax the ReLU constraint and arrive at a convex optimization problem. For this, we can enclose $\vect z$ inside a convex envelope determined by the bounds $\bm \ell, \vect u$ \cite{wong2018provable}:
\begin{equation}
    \begin{aligned}
        &\vect z,
        &\vect z \geq \hat{\vect{z}},\quad
        &(\vect{u} - \bm{\ell})\odot \vect{z} - \vect u \odot \hat{\vect{z}} \leq -\vect{u} \odot\bm{\ell}.\label{eq:convex_relaxation}
    \end{aligned}
\end{equation}
The relaxation of the ReLU constraint results in a convex optimization problem and can be solved rather efficiently but leads to an outer approximation of the adversarial polytope. This leads to a lower bound of the objective of \autoref{eq:general_verification}:
\begin{equation}
\underline{\vect{z}}^{[L]}_c - \overline{\vect{z}}^{[L]}_t \leq  {\vect{z}^*_c}^{[L]} - {\vect{z}_t^*}^{[L]}. 
\end{equation}

Here, ${\vect{z}^*}^{[L]}$ denotes the optimal value as given by the exact formulation and $\underline{\vect{z}}^{[L]}, \overline{\vect{z}}^{[L]}$ are the lower and upper bounds respectively. 

\textbf{Quadratic Unconstrained Binary Optimization (QUBO)}
Many problems in finance, economics and machine learning may be formulated as QUBO problems~\cite{QUBO}:
\begin{equation}
    \min_{\vect{x}\in\left\{0, 1\right\}^n} \vect{x}^\intercal \matr{Q} \vect{x} + \vect{q}\vect{x},
\end{equation}
where $\vect{x}$ is vector of binary variables, $\matr{Q} \in \mathbb{R}^{n\times n}$ is a square matrix and $\vect{q}\in\mathbb{R}^n$ is a row vector.

In QC, the QUBO formulation gained lots of attention due to its close resemblance to the Ising model~\cite{lucas2014ising}, a physical model directly linked to spin states, with main difference being that states take values in $\{-1, 1\}$ in Ising and $\{0,1\}$ in QUBO.
So, to get from an Ising model to the corresponding QUBO, one has to transform the states by $s=2x-1$ and adjust $\vect{h}$, and $\matr{J}$ accordingly:
\begin{equation}
     \min_{\vect{s}\in\left\{-1, 1\right\}^n} -\sum_{i}\vect{h}_i \vect{s}_i - \sum_{ij} \matr{J}_{ij}\vect{s}_i\vect{s}_j.
\end{equation}

\textbf{Quantum Annealing (QA)} aims at directly solving such Ising problems on quantum hardware~\cite{johnsonQuantumAnnealingManufactured2011} with the only restriction being the connectivity of the qubits on QA.
While in simulation and theory we may connect any qubit to any other qubit, this is not possible on current hardware.
To circumvent this issue, one may represent one logical unit $s_i$ by multiple qubits with large coupling weights $\matr{J}_{ij}$ in between~\cite{adachiApplicationQuantumAnnealing2015}.
While proven to be a suboptimal choice~\cite{marshallPerilsEmbeddingSampling2020} it remains an open research question to find suitable strategies of overcoming these limits of current hardware.

%Due to the mapping of QUBO problems to Ising models and thus Ising Hamiltonians, the solution of QUBO problems by QA essentially comes down to finding a ground state of the Ising model/Hamiltonian. To achieve this, the QA systems is first initialized in a known ground state and is then subsequently evolved towards the ground state of the problem Hamiltonian that encodes the original QUBO problem. If this evolution happens slowly enough, the adiabatic theorem states that the systems remains within a ground state during the whole evolution.\gao{I think such an explanation is valuable, though, we may want to adapt the terminology.}

Thus, we can obtain solutions for the QUBO problem by finding the ground state of the Ising model. In QA, we achieve this by starting from a known initial state which we slowly evolve towards the ground state of the problem. For a detailed description we refer the reader to~\cite{dasQuantumAnnealingRelated2005}.

\textbf{Quantum Approximate Optimization Algorithm (QAOA)}~\cite{farhiQuantumApproximateOptimization2014} is a hybrid quantum method for solving Ising problems. The algorithm is considered as an excellent candidate for NISQ devices due to its hybrid structure with an iteration between conventional computers and QC. Its potential in achieving a practical quantum advantage, however, remains unclear and is intensively studied~\cite{farhi2019quantum,QAOAhundreds}. In comparison to QA, the use of the QAOA on gate-based QC may open up the possibility to consider more complicated Hamiltonians, also going beyond the QUBO problem. The QAOA transforms the problem Hamiltonian into a sequence of local Hamiltonians and evaluates approximate ground states to obtain an approximate solution of the original optimization problem.

%% file: sections/method/overview.tex
\section{HQ-CRAN} \label{sec:method}

In this section, we provide our main contribution. Given the exponential complexity of a complete verification method, we highlight the potential of Benders decomposition to exploit NISQ devices and run part of the problem on QC while generating constraints on a classical one. Benders decomposition is an algorithm for mixed-variables programming that has been successfully applied to a variety of applications \cite{benders1962}. Recently, \citet{chang2020quantum} and \citet{zhao2021hybrid} proposed its use on QC. Following their approach, we present the first application to the verification problem of neural networks and supplement it with QC-specific improvements.%\tom{maybe stress again the first usage for verification?} 
In a simulation with a sufficient number of qubits, our algorithm theoretically converges to the value of the exact solution \cite{benders1962}.%

We consider a problem with a complicating variable $\vect y$, e.g., a binary variable in our case, rendering our problem to be an instance of the class of MIP problems. The main components of our method are two programs, one LP and one QUBO. These are alternatively solved until convergence or another stopping criterion is met:

\begin{enumerate}
    \item Solve LP with a given variable binary $\vect{y}$ which provides us with an optimal objective and a new constraint.
    \item Add the constraint to the QUBO and compute the optimal value and $\vect{y}$ on the NISQ device.
    \item If not converged go to 1.
\end{enumerate}
When the objectives of both programs meet, we can stop the optimization and issue a certificate if this value is positive. Due to the hardware constraints, we have to use approximations within the QUBO program leading to a non exact certificate but rather a conservative bound. This means that our certificate is valid and in the ideal, simulated scenario, we can converge to the exact solution with enough iterations. 
However, in settings where the number of iterations or qubits is limited, HQ-CRAN is an incomplete verification method rendering HQ-CRAN as a mixture of complete and incomplete verification.

%% file: sections/method/benders.tex
\subsection{Benders Decomposition}
For a proper certificate, we have to test whether we can change the network's prediction into any other possible class as in \autoref{eq:general_verification}. However, for simpler notation, we now only consider the inner minimization problem, i.e. only testing the difference between the initial predicted class and one other. To avoid cluttered notation, we write the complete verification problem in a concise matrix formulation as:
%\tom{maybe say that this is also closer to Benders original notation?}
%
\begin{equation}\label{eq:setup}
\min_{\vect{z},\vect{y}} \left\{\vect{g}^{\intercal}\vect{z} 
    \,| \, \matr{A}\vect{z} + \matr{B}\vect{y} \geq \vect{b}, \; \matr{C}\vect{z} \geq \vect{d} \right\}.
\end{equation}
%
%\tom{here the matrices A, B, ... are not known yet. need to be introduced here already. Maybe before the definition of stable neurons. otherwise it feels like we are going already to the next topic}
The matrices and vectors $\matr{A}$, $\matr{B}$, $\matr{C}$, $\vect{b}$, and $\vect{d}$ model the whole network and the exact construction procedure is explained in \autoref{alg:matrices_construction}, where we employ interval arithmetic to propagate the boundaries through the network \cite{tjeng2018evaluating}. The vector $\vect{z} \in \mathbb{R}^{n_z}$ considers all network logits with $n_z = n_0 + \dots + n_L$ equal to the total number of neurons, and $\vect{y} \in \left\{0, 1 \right\}^{n_y}$ with $n_y$ equal to the number of unstable neurons (i.e. a neuron is $unstable$ when its boundaries are $\vectsym{\ell}^{[i]}_j < 0 < \vect{u}^{[i]}_j$; otherwise we can express it with a constant being either 0 or 1). The vector $\vect{g}$ encodes the difference between the true class and another to be tested against, e.g. $\vect{g} = \left(\dots, 0, 1, 0, -1, 0\right)^\intercal$. %\tom{do we have an alternative on how we compose the matrices? this is a bit hard to read}

\begin{algorithm} 
\caption{Generation of problem matrices}
  \label{alg:matrices_construction}
   \begin{algorithmic}
   \Function{Build}{$\vect{x}, \epsilon, \matr{W}, \vect{v}$}
    \State \textbf{initialize:} $n_y \leftarrow 0$, $c\leftarrow \argmaxinline_j\vect{f}(\vect{x})_j$, $\vect{g} \leftarrow \vect{0}_{n_z}$
    \State $\vectsym{\ell}^{[0]},\, \vect{u}^{[0]}  \leftarrow \vect{x} - \epsilon \cdot \mathds{1}_{n_0},\, \vect{x} + \epsilon \cdot \mathds{1}_{n_0}$
    \For{$(\matr{W}^{[i]}, \vect{v}^{[i]})$ \textbf{in} $(\matr{W}, \vect{v})$} \Comment{where $^{[i]}$ is omitted}
        \State $\vectsym{\ell}, \vect{u} \leftarrow $ computed with Interval Arithmetic~\cite{tjeng2018evaluating}
        \State $\matr{M} \leftarrow (..., \mathbf{O}_{n_i\times n_{i-2}}, -\matr{W}, \mathbf{I}_{n_i}, \mathbf{O}_{n_i\times n_{i+1}}, ...)^\intercal$
        \State $\matr{N} \leftarrow (..., \mathbf{O}_{n_i\times n_{i-1}}, \mathbf{I}_{n_i}, \mathbf{O}_{n_i\times n_{i+1}}, ...)^\intercal$
        \For{$(u_j, \ell_j)$ \textbf{in} $(\vect{u}, \vectsym{\ell})$}
            \If {$\ell_j \geq 0$ } \Comment{stable active}
                \State $\matr{\hat{A}}_{:, j},\,\,\, \vect{\hat{b}}_j,\,\,\, \vect{\check{b}}_j \leftarrow -{\matr{M}_{:, j}}, -\vect{v}_j, -u_j$
                \State $\matr{\hat{B}}_{:, j},\,\,\, \matr{\check{B}}_{:, j} \leftarrow \vect{0}_{n_y},\,  \vect{0}_{n_y} $
            \ElsIf {$u_j \leq 0$ } \Comment{stable inactive}
                \State $\matr{\hat{A}}_{:, j},\,\,\, \vect{\hat{b}}_j,\,\,\,\vect{\check{b}}_j \leftarrow \vect{0}_{n_z},  -\vect{v}_j,\, 0$
                \State $\matr{\hat{B}}_{:, j},\,\,\, \matr{\check{B}}_{:, j} \leftarrow \vect{0}_{n_z},\,\,\, \vect{0}_{n_y}$
            \Else \Comment{unstable neuron}
                \State $\matr{\hat{A}}_{:, j},\,\,\, \vect{\hat{b}}_j,\,\,\, \vect{\check{b}}_j \leftarrow -{\matr{M}_{:, j}},\, 0,\, -\vect{v}_j$
                \State $\matr{\hat{B}}_{:, j},\,\,\, \matr{\check{B}}_{:, j} \leftarrow (0\,..\,\ell_j\,..\, 0)^\intercal,\, (0\,..\,u_j\,..\, 0)^\intercal$
                \State $n_y \leftarrow n_y + 1$
            \EndIf
        \EndFor
    \EndFor
    \State $\matr{A} \leftarrow (\matr{\hat{A}}^{[1]}, ..., \matr{\hat{A}}^{[L]},-\matr{N}^{{[1]}},...,-\matr{N}^{{[L]}})^\intercal$
    \State $\matr{B} \leftarrow (\matr{\hat{B}}^{[1]}, ..., \matr{\hat{B}}^{[L]},\matr{\check{B}}^{{[1]}},...,\matr{\check{B}}^{{[L]}})^\intercal$
    \State $\matr{T} \leftarrow (\mathbf{I}_{n_0}, \mathbf{O}_{n_0\times(n_1+\dots+n_L)})^\intercal$
    \State $\matr{C} \leftarrow (\matr{T}, -\matr{T}, \matr{M}^{{[1]}}, ..., \matr{M}^{{[L]}}, \matr{N}^{{[1]}},... \matr{N}^{{[L]}})^\intercal$
    \State $\vect{b} \leftarrow (\hat{\vect{b}}^{[1]}, ...,\hat{\vect{b}}^{[L]},\check{\vect{b}}^{[1]},...,\check{\vect{b}}^{[L]})^\intercal$
    \State $\vect{d} \leftarrow (\vectsym{\ell}^{[0]}, -\vect{u}^{[0]}, \vect{v}, \vect{0}_{n_1+\dots+n_L})^\intercal$
    \State $\vect{g}_{n_z - c} \leftarrow 1$
    \State \Return $\matr{A}, \matr{B}, \matr{C}, \vect{b}, \vect{d}, \vect{g}$
\EndFunction
\end{algorithmic}
\end{algorithm}

% At this point, for consistent and comprehensible reading, we follow similar steps as in \citet{chang2020quantum} while having a different initial problem. 
We decompose \autoref{eq:setup} into two sub-problems: A binary unconstrained problem and a linear programming optimization problem. Let us rewrite \autoref{eq:setup} as $\textstyle\min_{\vect{y}} q(\vect{y})$ where: 
\begin{equation}\label{eq:reformulated}
    q(\vect{y}) = \min_{\vect{z}} \left\{\vect{g}^{\intercal}\vect{z}
   \,| \, \matr{A}\vect{z} + \matr{B}\vect{y} \geq \vect{b}, \; \matr{C}\vect{z} \geq \vect{d} \right\}.
\end{equation}
Here we view the vector of binary variables $\vect y$ as given. Hence, we decoupled $\vect y$ from the rest of the program resulting in a LP. The dual formulation of $q(\vect{y})$ is given as:
\begin{equation}\label{eq:dual}
    \max_{\vectsym{\alpha}, \vectsym{\beta}} \; \left\{\vectsym{\alpha}\left( \vect{b} - \vect{B} \vect{y} \right) + \vectsym{\beta} \vect{d} \; | \; \vectsym{\alpha} \matr{A} + \vectsym{\beta} \matr{C} = \vect{g}^\intercal \right\},
\end{equation}
where $\vectsym{\alpha}\in \mathbb{R}^{m_b}_+$ and $\vectsym{\beta}\in \mathbb{R}^{m_d}_+$ are row vectors.
% \begin{equation}\label{eq:reformulated-dual}
%     \max_{\vectsym{\alpha}\in \mathbb{R}^{m_b}_+, \vectsym{\beta} \in \mathbb{R}^{m_d}_+} \, \inf_{\vect{z}} \mathscr{L} \left( \vectsym{\alpha}, \vectsym{\beta} \right),
% \end{equation}
% where $\mathscr{L}(\vectsym{\alpha}, \vectsym{\beta})$ is the Lagrangian function given by:
% \begin{equation}
%     \mathscr{L}(\vectsym{\alpha}, \vectsym{\beta}) = \vect{g}^\intercal \vect{z} - \vectsym{\alpha}^\intercal \left( \matr{A} \vect{z} + \matr{B} \vect{y} - \vect{b} \right) - \vectsym{\beta}^\intercal \left( \matr{C} \vect{z} - \vect{d} \right).
% \end{equation}
Since $\vect{y}$ is constant within the optimization of \autoref{eq:reformulated}, the optimization program is a linear program and we thus have strong duality\footnote{i.e. the optimal objective value of the primal equals the optimal value of the dual.} %\tom{i think we can omit this footnote. it should be basic knowledge for the audience? }. 
The optimal objective value of \autoref{eq:dual} is infinity if \autoref{eq:reformulated} is infeasible. If \autoref{eq:reformulated} is feasible it has to have the same finite value. Thus, it should be either finite or positive infinity. %\tom{do we need to cite chang here again?}
Therefore, for an optimal value of the dual formulation of \autoref{eq:dual}, $\vect{g}^\intercal - \vectsym{\alpha}\matr{A} - \vectsym{\beta}^\intercal \matr{C}$ has to be zero. 
%because otherwise, taking the infimum over all $\vect{z}$ will result in negative infinity. 
Hence, we know the optimal feasible objective will only be the result of the remaining terms that are not interacting with $\vect{z}$ as all of them need to cancel out. 
% So we can reformulate the dual as\footnote{See \citet{chang2020quantum} for more details.}:
Following Benders decomposition, we can formulate the objective of \autoref{eq:dual} as a linear combination of extreme rays and points of the feasible region. We denote as $\Lambda_r$ and $\Lambda_p$ the set of extreme rays and extreme points of the set $\left\{ (\vectsym{\alpha}, \vectsym{\beta}) \ | \  \vectsym{\alpha} \matr{A} + \vectsym{\beta} \matr{C} = \vect{g}^\intercal, \vectsym{\alpha} \geq 0, \vectsym{\beta} \geq 0 \right\}$. 

Similarly to \citet{benders1962}, we rewrite the optimization of \autoref{eq:setup} to obtain:
\begin{subequations} \label{eq:master_problem}
\begin{alignat}{2} 
    &\min_{\vect{y},\, \eta} \, \eta, \\
    \text{s.t.} \
    &\vectsym{\alpha}^{[k]}\left( \vect{b} - \matr{B} \vect{y} \right) + \vectsym{\beta}^{[k]}\vect{d}\leq \eta, \; \forall (\vectsym{\alpha}^{[k]}, \vectsym{\beta}^{[k]}) \in \Lambda_p, \label{eq:master_extreme_point}\\
    &\vectsym{\alpha}^{[k]}\left( \vect{b} - \matr{B} \vect{y} \right) + \vectsym{\beta}^{[k]} \vect{d} \leq 0, \; \forall (\vectsym{\alpha}^{[k]}, \vectsym{\beta}^{[k]}) \in \Lambda_r, \label{eq:master_extreme_ray}
\end{alignat}
\end{subequations}
where $\eta\in\mathbb{R}$ is a scalar. This problem is known as \textit{master} problem on Benders decomposition. 

% \tom{we should think about pulling this theorem to the front of the benders decomposition section. and put therest into the proof? not sure}
\begin{theorem}\label{thm:soundness}
Given a neural network $\vect f$, an input $\vect{x}$ and logits $\vect{z}$ and following \autoref{alg:matrices_construction} to create constraint matrices $\matr{A}, \matr{B}, \matr{C}$ as well as vectors $\vect{b}, \vect{d}, \vect{g}$, the verification of robustness can be evaluated through the optimization of \autoref{eq:master_problem}.
\end{theorem}
\input{sections/method/proof}

% We state this theorem here to express that the optimization of the reformulation is sound and a certificate issued by this procedure is valid. The proof of \autoref{thm:soundness} uses the decomposition theorem of \citet{benders1962} and is given in \autoref{app:theorems}. 

The difficulty of solving \autoref{eq:master_problem} is the exponential size of the sets $\Lambda_p, \Lambda_r$ \cite{chang2020quantum}. Thus, we can gradually extend the sets $\Lambda^\prime_p \subseteq \Lambda_p, \Lambda^\prime_r \subseteq \Lambda_r$ by constraints of the \textit{sub} problem defined as:
\begin{equation} \label{eq:sub_problem}
        \max_{\vectsym{\alpha} \leq \bar{\vectsym{\alpha}},\, \vectsym{\beta} \leq \bar{\vectsym{\beta}}} 
        \left\{\vectsym{\alpha}\left( \vect{b} - \matr{B} \vect{y} \right) + \vectsym{\beta} \vect{d} \ | \
        \vectsym{\alpha} \matr{A} + \vectsym{\beta} \matr{C} = \vect{g}^\intercal \right\}.
\end{equation}
The sub problem is similar to \autoref{eq:dual} except that $\vectsym{\alpha}$ and $\vectsym{\beta}$ are bounded. This bounding idea, introduced by \citet{chang2020quantum}, can be used to identify whether we add a cut belonging to the extreme points or rays. If any of the optimal values of the vectors $\vectsym{\alpha}$ or $\vectsym{\beta}$ is equal to the upper bounds, the constraint belongs to the extreme rays $\Lambda^\prime_r$. Otherwise, it belongs to the extreme points $\Lambda_p^\prime$. As the set of extreme points is bounded, this decision rule is correct as long as the bounds $\bar{\vectsym{\alpha}}, \bar{\vectsym{\beta}}$ are sufficiently large, i.e. larger than any solution of the extreme points set \cite{chang2020quantum}.

%% file: sections/method/proof.tex
\begin{proof}
As shown in \autoref{alg:matrices_construction}, we can equivalently write the initial problems in terms of the matrices $\matr{A}, \matr{B}, \matr{C}$ as well as vectors $\vect{b}, \vect{d}$ and $\vect{g}$. Thus, it remains to show the proper application of Benders theorem. 
%We start with our initial problem reformulation \autoref{eq:general_verification} and reformulate it to fit the Partitioning theorem for mixed-variables programming problems \cite{benders1962}. 
% We start with restating our initial robustness verification here: 
% \begin{subequations}
% \begin{alignat}{2}
% &\min_{\vect{z} \in \mathbb{R}^{n},\vect{y} \in \left\{0, 1\right\}^{n_y}} &\quad&\vect{g}^{\intercal}\vect{z} \\
%     &\qquad\,\,\,\text{subject to} &&\matr{A}\vect{z} + \matr{B}\vect{y} \geq \vect{b} \\
%     &&&\matr{C}\vect{z} \geq \vect{d}. 
% \end{alignat}
% \end{subequations}
% The dimensions are: $\vect g, \vect b, \vect z \in \mathbb{R}^n, d \in \mathbb{R}^p, \matr{A} \in \mathbb{R}^{m \times n}, \matr C \in \mathbb{R}^{p \times n}, \matr{B}\vect y \in \mathbb{R}^m$. 
We start by rewriting \autoref{eq:setup} to a maximization problem by flipping the objective. Also multiply the constraints with -1 and absorb the minus sign in front of a parameter within the new parameters $\hat{\matr{A}}, \hat{\matr{B}}, \hat{\matr{C}}, \hat{\vect{b}}, \hat{\vect{d}}, \hat{\vect{g}}$: 
\begin{equation}
\textstyle\max_{\vect{z},\vect{y}} \{\hat{\vect{g}}^{\intercal}\vect{z}\, |\,  
\hat{\matr{A}}\vect{z} + \hat{\matr{B}}\vect{y} \leq \hat{\vect{b}}, \,\hat{\matr{C}}\vect{z} \leq \hat{\vect{d}}, \, \vect{z}\in\mathbb{R}^{n_z}\} 
\end{equation}
Now, we change the constraints to fit the form: $\Tilde{\matr{A}}\vect{z} - \vect{F}(\vect{y}) \leq \Tilde{\vect b}$:
\begin{equation*}
	\begin{aligned}
\Tilde{\matr{A}} = \left[ {\begin{array}{cc}
    \hat{\matr{A}} \\
    \hat{\matr{C}} \\
  \end{array} } \right], \,
\vect{F}(\vect{y}) = \left[ {\begin{array}{c}
    \hat{\matr{B}}\vect{y}\\
    \matr{0}\\
  \end{array} } \right], \,
\Tilde{\vect{b}} = \left({\begin{array}{c}
     \hat{\vect{b}}\\
     \hat{\vect{d}}
\end{array}}\right), \,
\Tilde{\vect{g}} = \left({\begin{array}{c}
     \hat{\vect{g}}\\
     \vect{0}
\end{array}}\right),
  \end{aligned}
\end{equation*}
As we want to have our optimization variables constrained to be positive, we express $\vect z$ as the difference of two positive vectors. This can equivalently be expressed with the matrix $\vect M$ constituting of two identity matrices. Formally, we write: 
\begin{equation}
\begin{aligned}
    \vect z = \vect{x}_1 - \vect{x}_2
    = \begin{bmatrix}
        \mathbf{I}_m & \mathbf{0}_m \\
        \mathbf{0}_m & -\mathbf{I}_m
    \end{bmatrix} \vect x 
    \coloneqq \matr{M} \vect x
\end{aligned}
\end{equation}
Inserting the reformulations into our program, we get:
\begin{equation}
\textstyle\max_{\vect{x},\vect{y}}\{\hat{\vect{g}}^{\intercal}\matr{M}\vect{x} \, | \, \Tilde{\matr{A}}\matr{M}\vect{x} + \hat{\matr{B}}\vect{y} \leq \hat{\vect{b}}, \, \vect x \geq \vect{0},\, \vect{x} \in \mathbb{R}^{2n_z}_+\}
\end{equation}
Defining $\Bar{\matr{A}} \coloneqq \Tilde{\matr{A}}\matr{M}, \vect c \coloneqq \matr{M}^\intercal \Tilde{\vect{g}}, f(\vect{y}) = 0$ leads us to the desired form of problems directly suited for Benders theorem \cite{benders1962}:
\begin{equation} \label{eq:initial-bender}
    \max\{\vect{c}^\intercal \vect{x}+ f(\vect y)\, |\, \Tilde{\matr{A}} +  \vect{F}(\vect{y}) \leq \Tilde{\vect{b}}, \vect x \in \mathbb{R}_+^q, \vect y \in \set{S}\}.
\end{equation}

We can now decompose this problem using the 3.1 Partitioning theorem for mixed-variables programming problem \cite{benders1962} and receive two problems, i.e. \textit{master} and \textit{sub} problem. The \textit{master} is given as follows:
\begin{equation} \label{eq:master-bender}
\begin{aligned}
    &\max\{x_0 | (x_0, \vect{y}) \in \set{G}\}, \\
    &\set{G} = \bigcap_{(u_0, \vect{u}) \in \set{C}}\{(x_0, \vect{y})\, |\, u_0 x_0 + \vect{u}^\intercal \vect{F}(\vect{y}) - u_0 f(y) \leq \vect{u}^\intercal \Tilde{\vect{b}}\},
\end{aligned}
\end{equation}
where $\set{C} = \{(u_0, \vect{u}) | \Bar{\matr{A}}^\intercal\vect{u} - \vect c u_0 \geq \vect 0, \vect{u} \geq \vect{0}, u_0 \geq 0\}$. 
The \textit{sub} problem with a fixed $\vect y$ is given by:
\begin{equation}\label{eq:sub-bender}
\begin{aligned}
    \max\{\vect c^\intercal \vect x\, |\, \Bar{\matr{A}} \leq \vect b - \vect F(\vect y), \ \vect x \geq \vect 0\}.
\end{aligned}
\end{equation}
According to the theorem (i) \autoref{eq:initial-bender} is infeasible if and only if \autoref{eq:master-bender} is infeasible. Similarly, (ii) \autoref{eq:initial-bender} is feasible without optimum solution if and only if \autoref{eq:master-bender} is feasible without optimum solution. (iii) If $(x^*, y^*)$ is an optimum solution to initial problem \autoref{eq:initial-bender} and $x_0^*=\vect{c}^\intercal \vect{x}^* + f(y^*)$, then $(x_0^*, y^*)$ is an optimum solution sub problem \autoref{eq:sub-bender}. Lastly, (iv), if $(x_0^*, y^*)$ is an optimum solution to the \textit{master} problem \autoref{eq:master-bender}, then the sub-problem \autoref{eq:sub-bender} is feasible with the optimum value equal to $x^* = x_0^* - f(y)$. If $x^*$ is an optimum solution to the sub problem \autoref{eq:sub-bender}, then $(x^*, y^*)$ is also an optimum solution to the initial formulation, which in our case means $x^* = x_0^*$ as we model $f(y)=0$.

%Thus, this theorem shows us how to decompose the initial formulation into the \textit{master} and sub problem. 
As the computation of the \textit{master} over the entire set $\set{G}$ is usually not feasible in practice, \citet{benders1962} furthermore proposes two Lemmas -- 4.1 and 4.2 -- to add constraints to the set $\set G$ in an iterative way. Given the sets:
\begin{equation}
    \set{C}_0 = \{\vect{u} | \Bar{\matr{A}}^\intercal \vect u \geq \vect 0, \vect u \geq \vect 0\}, \quad \set{P} = \{\vect u\ | \Bar{\matr{A}}^\intercal \vect u \geq \vect 0 \},
\end{equation}
a point $(\vect x, \vect y)$ is contained in $\set G$ if and only if: 
\begin{subequations}
\begin{alignat}{2}
 (b-\vect F(\vect y))^\intercal \vect u \geq 0, \forall \vect{u} \in \set{C}_0\\
 (b-\vect F(\vect y))^\intercal \vect u + f(\vect{y}) \geq x_0, \forall \vect{u} \in \set{P}.
\end{alignat}
\end{subequations}
Now, we can follow the classic iterative approach from \citet{benders1962} to get the optimum solution by iteratively computing the \textit{master} and \textit{sub} problem. To finish our proof of validity for the computation of the certificate, we will next show that when reversing the transformations for the \textit{master} problem, \textit{sub} problem and the sets, we arrive at the proposed problems.

Renaming $\vect u = \left(\vectsym{\alpha}\,\, \vectsym{\beta}\right)$ and inserting the expressions for $\Bar{\matr{A}}, \vect c, \vect{F}(\vect{y})$ and $\Tilde{\vect{b}}$ into the \textit{master} problem yields:
\begin{equation} \label{eq:master-bender2}
\begin{aligned}
    &\max\{x_0\, |\, (x_0, \vect{y}) \in \set{G}\},\\
    &\set{G} = \textstyle\bigcap_{(u_0, \vect{u}^\intercal)\in\set{C}}\{(x_0,\vect{y})\,|\,\vect{u}(\matr{F}(\vect{y}) - \tilde{\vect{b}}) \leq -u_0 x_0\}.
\end{aligned}
\end{equation}
As we had our initial formulation as minimization problem, we rephrase the above into one by flipping the objective and define $\eta \coloneqq -x_0$. 
The resulting problem is:
\begin{equation}\label{eq:bender-proof-master}
\begin{aligned}
    &\min\{\eta \, | \, (\eta, \vect{y}) \in \set{G}\},\\
    &\set{G} = \textstyle\bigcap_{(u_0, \vect{u}^\intercal)\in\set{C}}\{(\eta,\vect{y})\,|\,\vect{u}(\matr{F}(\vect{y}) - \tilde{\vect{b}}) \leq u_0 \eta\}.
\end{aligned}
\end{equation}
Again, a point $(\vect x, \vect y)$ is contained in $\set{G}$ if and only if:
\begin{subequations}\label{eq:bender-proof-master-constraints}
\begin{alignat}{2}
 (b-\vect F(\vect y))^\intercal (\vectsym \alpha\,\, \vectsym \beta)^\intercal \geq 0, \forall (\vectsym \alpha\,\, \vectsym \beta)^\intercal \in \set{C}_0\\
 (b-\vect F(\vect y))^\intercal (\vectsym \alpha\,\, \vectsym \beta)^\intercal + f(\vect{y}) \geq -\eta, \forall (\vectsym \alpha\,\, \vectsym \beta)^\intercal \in \set{P}.
\end{alignat}
\end{subequations}

Now changing the optimization problem back into a minimization program and combining \autoref{eq:bender-proof-master} with \autoref{eq:bender-proof-master-constraints} yields:
\begin{subequations}
\begin{alignat}{2} 
    &\min_{\vect{y},\, \eta \in \mathbb{R}} \, \eta, \\
    \text{s.t.} \;
    &\vectsym{\alpha}\left( \vect{b} - \matr{B} \vect{y} \right) + \vectsym{\beta}\vect{d}\leq \eta,\quad \forall (\vectsym{\alpha},\vectsym{\beta})\in \set{C}_0,\\
    &\vectsym{\alpha}\left( \vect{b} - \matr{B} \vect{y} \right) + \vectsym{\beta}\vect{d}\leq 0,\quad \forall (\vectsym{\alpha},\vectsym{\beta})\in \set{P}.
\end{alignat}
\end{subequations}
In the same way, we can transform the \textit{sub} problem and the constraint sets back and arrive at \autoref{eq:dual}.
\end{proof}

%% file: sections/method/improvements.tex
\subsection{Improvements}\label{sec:improvements}

As highlighted in \cite{magnanti1981accelerating, ruszczynski1997accelerating, papadakos2008practical}, several ways to improve Benders decomposition exist, including but not limited to (i) a smart selection of cuts and (ii) guiding the choice of $\vect{y}^{[t]}$ at every iteration. 
\citet{magnanti1981accelerating} were the first to note that cuts of different strengths can be generated from the \textit{sub} problem. In their work, they compared cuts according to their strength to obtain the strongest cut possible. Formally, given \autoref{eq:master_problem}, the cut $\vectsym{\alpha}^{[1]}\left(\vect{b} - \matr{B}\vect{y}\right) + \vectsym{\beta}^{[1]}\vect{d}$ \textit{dominates} or is \textit{stronger} than the cut $\vectsym{\alpha}^{[2]}\left(\vect{b} - \matr{B}\vect{y}\right) + \vectsym{\beta}^{[2]}\vect{d}$, if $\vectsym{\alpha}^{[1]}\left(\vect{b} - \matr{B}\vect{y}\right) + \vectsym{\beta}^{[1]} \vect{d} \geq \vectsym{\alpha}^{[2]}\left(\vect{b} - \matr{B}\vect{y}\right) + \vectsym{\beta}^{[2]}\vect{d}$ for all $\vect{y} \in \mathcal{Y}$ with a strict inequality for at least one point $\vect{y}$. Hence, it follows the definition of a \textit{Pareto-optimal} cut.
\begin{definition}[Pareto-optimal~\cite{magnanti1981accelerating}]
A cut is \textit{Pareto-optimal} if it is not dominated by any other cut. 
\end{definition}

To compute a Pareto-optimal cut, \citet{magnanti1981accelerating} integrated to the standard Benders decomposition an additional problem and introduced the notion of core point. A \textit{core point} of $\mathcal{Y}$ is any point $\bar{\vect{y}}^{[t]}$ contained in the relative interior $ri\left(\mathcal{Y}^c\right)$ of the convex hull $\mathcal{Y}^c$ of $\mathcal{Y}$. For further insights of the connection between a core point and a Pareto-optimal cut, we refer the reader to \citet{magnanti1981accelerating}.
%\tom{here we are missing the connection between core point and pareto-optimal cut}
The results obtained from Pareto-optimal cuts demonstrated a significant improvement in the convergence of the algorithm. However, finding a good strategy to estimate a core point is not straightforward and limits the use of the Magnanti-Wong method. To further accelerate the solution of this additional problem, \citet{papadakos2008practical} showed how to generate a Pareto-optimal cut with an independent additional problem.
%
%\tom{not sure if you mixed the citations here: you introduce AP from Magnanti Wong 5 and then said Papadakos 6 showed how to generate them and then now you deine it from Magnanti Wong 6}
\begin{theorem}[Independent Magnanti-Wong~\cite{papadakos2008practical}]
Let $\bar{\vect{y}}^{[t]}$ be a core point of $\mathcal{Y}$, then the optimal solution $\left(\vectsym{\alpha}^{[t]}, \vectsym{\beta}^{[t]}\right)$ of the following problem:
\begin{equation}\label{eq:additional_problem}
    \textstyle\max_{\vectsym{\alpha}, \vectsym{\beta}} \; \left\{\vectsym{\alpha}\left(\vect{b} - \matr{B}\bar{\vect{y}}^{[t]}\right) + \vectsym{\beta}\vect{d} \; |\; \vectsym{\alpha}\matr{A} + \vectsym{\beta}\matr{C} = \vect{g}^\intercal \right\}
\end{equation}
is Pareto-optimal.
\end{theorem}
We refer to \autoref{eq:additional_problem} as \textit{additional} problem.
According to \citet{papadakos2008practical}, $\bar{\vect{y}}^{[t]}$ does not have to be a core point of $\mathcal{Y}$ to give a Pareto-optimal cut, but any linear combination, such as $\bar{\vect{y}}^{[t]} \leftarrow \frac{1}{2} \bar{\vect{y}}^{[t]} + \frac{1}{2}\vect{y}^{[t]}$, where $\vect{y}^{[t]}$ is the solution of the \textit{master} problem at every iteration. With this approach, $\bar{\vect{y}}^{[t]}$ gradually approaches a core point \cite{papadakos2008practical}. 

The next improvement comes from the sensitivity of the \textit{master} to changes in the sets of cuts: $\Lambda^\prime_p, \Lambda^\prime_r$. We saw that the initial iterations may be inefficient and the \textit{master} shows an unstable behaviour when it is close to the solution, also called \textit{tailing effect} \cite{ruszczynski1997accelerating}. To overcome these difficulties, \citet{ruszczynski1997accelerating} propose to add an additional quadratic regularizing term to the \textit{master} objective:
\begin{equation}\label{eq:hamming}
    \textstyle\min_{\eta, \vect{y}} \; \eta + \frac{1}{2}\norm{\vect{y} - \vect{y}^{[t-1]}}_2^{2} \quad \text{s.t.} \; \left(\ref{eq:master_extreme_point}\right) - \left(\ref{eq:master_extreme_ray}\right) \; \text{hold},  
\end{equation}
where $\vect{y}^{[t-1]}$ is the solution at the previous iteration and $\norm{\cdot}_2$ is the Euclidian norm. Given the binary nature of $\vect{y}$, we will refer to this term as Hamming distance. It penalizes the current solution to jump far away from the previous one and makes it possible to remove inactive cuts and preserve the size of $\Lambda^\prime_p$ and $\Lambda^\prime_r$. 
%\tom{how can we preserve the cuts?}

%% file: sections/method/bridge.tex
\subsection{Bridging to QC}\label{sec:adjustments}

As exposed in \autoref{sec:background}, the QUBO model is used as basis for many NISQ-optimization algorithms. Therefore, we focus on converting the \textit{master} into a QUBO formulation. To this end, we start by adding an artificial variable $a_k \in \mathbb{R}^+$ to relax the inequality and to penalize each constraint:\footnote{w.l.o.g, we consider only the set of optimal cuts $\Lambda^\prime_p$}
\begin{equation}\label{eq:bridge}
\begin{aligned}
    \min_{\vect{y}, \eta, a_1, \dots, a_\tau} &\eta + \frac{1}{2}\norm{\vect{y} - \vect{y}^{[t-1]}}^2 + \\
    \sum_{(\vectsym{\alpha}^{[k]},\vectsym{\beta}^{[k]})\in\Lambda_p^\prime}&\left(\vectsym{\alpha}^{[k]} \left(\vect{b}-\matr{B}\vect{y}\right)+\vectsym{\beta}^{[k]}\vect{d} - \eta + a_k\right)^2
\end{aligned}
\end{equation}
where $k\in\left\{1, \dots, \tau\right\}$ and $\tau = |\Lambda_p^\prime|$. 
Then, we approximate the real variables $\eta$ and $a_k$, with binary vectors $\vect{p}\in\left\{0, 1\right\}^{n_p}$ and $\vect{a}^{[k]}\in\left\{0, 1\right\}^{n_{a_k}}$, respectively. To this end, we used the fixed-point approximation, i.e the two's-complement number system encodes positive and negative numbers in a binary representation as:
\begin{subequations}
\begin{align}
\eta &\approx \vect{w}_p\vect{p} = \omega_p \cdot (- 2^{n_p-1}\vect{p}_{n_p-1}  + \sum_{i = 0}^{n_p - 2} 2^i \vect{p}_i ), \label{eq:eta} \\
a_k &\approx \vect{w}^{[k]}_a\vect{a}^{[k]} = \omega_a \cdot \sum_{i = 0}^{n_{a_k}} 2^i \vect{a}^{[k]}_i, \label{eq:a_k}
\end{align}
\end{subequations}
where $\omega_p, \omega_{a} \in \mathbb{R}^+, 0 < \omega_p, \omega_a \leq 1$ are the penality weights, $n_p, n_{a_k} \in \mathbb{N}$ are the number of qubits and $\vect{w}_p, \vect{w}_a^{[k]}$ are row vectors in $\mathbb{R}^{n_p}$ and $\mathbb{R}^{n_{a_k}}$, respectively. As noticed, $\vect{w}_{a}^{[k]} \vect{a}^{[k]}$ can only reach positive values, where $\vect{w}_p \vect{p}$ can also reach negatives. Given the limited number of qubits of the current gate-based QC, we relate $n_p$ to our objective function. 
\begin{corollary}\label{cor:p}
Let $\vect{u}^{[L]}_c, \vect{u}^{[L]}_t \in \mathbb{R}^{+}$ be the upper bounds in the last neural network layer of the prediction and target class, respectively, then the minimum number of qubits needed to approximate $\eta$ by a factor $\omega_p$ is
\[ n_p \geq 1 + \left\lceil \log_2 \left(1 + \frac{\vect{u}_c^{[L]} + \vect{u}_t^{[L]}}{\omega_p}\right) \right\rceil .\]
\end{corollary}
\begin{proof}
Let us consider the upper bounds constraint of the last layer for the ReLU formulation of \autoref{eq:complete_verification}: $\vect{z}^{[L]} \leq \vect{u}^{[L]} \odot \vect{y}^{[L]}$. The relative constraints for the predicted and target class are $\vect{z}_c \leq \vect{u}_c \cdot \vect{y}_c$ and $\vect{z}_t \leq \vect{u}_t \cdot \vect{y}_t$, respectively. Where the index of the last layer has been omitted for the sake of conciseness. Since we focus on containing the maximum attainable value of our objective, only the worst case is considered $\vect{y}_c = \vect{y}_t = 1$, i.e. both neurons are active.
So, we obtain $\vect{z}_c \leq \vect{u}_c$ and $\vect{z}_t \leq \vect{u}_t$. Adding $\vect{z}_t$ to $\vect{z}_c \leq \vect{u}_c$ results in $\vect{z}_c + \vect{z}_t \leq \vect{u}_c + \vect{z}_t \leq \vect{u}_c + \vect{u}_t$.
We also rely on the fact that $\vect{g}^\intercal \vect{z} = \vect{z}_c - \vect{z}_t \leq \vect{z}_c + \vect{z}_t$.

Since the objective of \autoref{eq:master_problem} is $\eta \leq \vect{g}^\intercal \vect{z}$ \cite{benders1962} at each iteration, we obtain:
\begin{equation*}
    \eta \leq \vect{g}^\intercal \vect{z} \leq \vect{u}_c + \vect{u}_t.
\end{equation*}
Lastly, by defining the maximum of $\vect{w}_p \vect{p}$ as $\bar{\eta} \coloneqq \omega_p \cdot \left(2^{n_p - 1} - 1 \right)$ and by imposing:
\begin{equation*}
    \vect{u}_c + \vect{u}_t \leq \omega_p \cdot \left(2^{n_p - 1} - 1 \right),
\end{equation*}
we obtain $\eta \leq \bar{\eta}.$
\end{proof}
To obtain a more compact notation, we introduce the row vector $\vect{h}^{[k]}\in \mathbb{R}^{n_t}$, defined as:
\begin{equation} \label{eq:cut}
    \vect{h}^{[k]} = 
    \begin{pmatrix}
        -\vect{w}_p &-\vectsym{\alpha}^{[k]}\matr{B} &\vect{0}^\intercal_{n_{a_1}} \;\dots\;\vect{w}_{a_k}\;\dots \;\vect{0}^\intercal_{n_{a_\tau}}
    \end{pmatrix},
\end{equation}
where $n_t = n_p + n_y + \dots +n_{a_\tau}$ is the total number of \textit{master} variables (or qubits) of the \textit{master} problem at iteration $t$. Finally, we define $\vect{x} = \left(\vect{p}, \vect{y}, \vect{a}^{[1]}, \dots, \vect{a}^{[\tau]}\right)$ as a column vector and we rewrite \autoref{eq:bridge} into a QUBO formulation:
\begin{equation}\label{eq:qubo}
    \min_\vect{x}\; \vect{x}^\intercal\matr{Q}^{[\tau]}\vect{x} + \vect{q}^{[\tau]}\vect{x} + \kappa^{[\tau]},
\end{equation}
with:
\begin{equation*}
    \begin{aligned}
    \vect{q}^{[\tau]} &= 2\sum_{k=1}^{\tau}e_k\cdot \vect{h}^{[k]} + \left(\vect{w}_p, - \vect{y}^{[t-1]},..., \vect{0}_{n_{a_\tau}}\right), \\
    \matr{Q}^{[\tau]} &= \sum_{k=1}^{\tau} \left({\vect{h}^{[k]}} \otimes \vect{h}^{[k]} \right) + \texttt{diag}\left(\vect{0}_{n_p}, \frac{1}{2}\cdot\mathds{1}_{n_y},..., \vect{0}_{n_{a_\tau}}\right), \\
    \kappa^{[\tau]} &= \sum_{k=1}^{\tau}e_k^2 + \frac{1}{2}\sum_{i = 0}^{n_y}\vect{y}^{[t-1]}_i,
    \end{aligned}
\end{equation*}
and $e_k = \vectsym{\alpha}^{[k]}\vect{b} + \vectsym{\beta}^{[k]}\vect{d}$ is a scalar. The vector $\vect{h}^{[k]}$ and the scalar $e_k$ express the cut generated from the \textit{sub} problem. Without any limitation to the size of the set of cuts $\Lambda^\prime_p$, at each iteration, the number of qubits increases by $n_{a_t}$. Therefore after few iterations \autoref{eq:qubo} reaches already a very large search space. To bound $n_{a_k}$, we relate it in a similar fashion as Corollary~\ref{cor:p} to the maximum reachable value by the cut.
\begin{algorithm}[tb]
   \caption{HQ-CRAN}
   \label{alg:hqcran}
\begin{algorithmic}
    \State\hspace{-0.5cm} \textbf{input:} $\vect{x},\epsilon,\matr{W},\vect{v}, T, \xi, \varphi, \bar{\vectsym{\alpha}},\bar{\vectsym{\beta}}, \omega_p, \omega_a$
    \State\hspace{-0.5cm} \textbf{output:} robust or unknown
    \State\hspace{-0.5cm} \textbf{initialize:} $\matr{A},\matr{B},\matr{C},\vect{d},\vect{b},\vect{g} \leftarrow$ \textsc{Build$(\vect{x}, \epsilon, \matr{W}, \vect{v})$}.
    \State set $c\leftarrow \argmaxinline_j\vect{f}(\vect{x})_j$
    \For{$j$ \textbf{in} $\mathcal{Z}\setminus \left\{c\right\}$}
        \State set $\vect{g}_{n_z - j} \leftarrow -1$, $\vect{y} \leftarrow \vect{0}_{n_y}$, $\bar{\vect{y}} \leftarrow \vect{0}_{n_y}$ and $t \leftarrow 1$
        \State $n_p $ computed according to Corollary~\ref{cor:p}
        \For{$t$ \textbf{in} $\{1, \dots, T\}$} 
            \State $s_{t} \leftarrow \argmax$ of \autoref{eq:sub_problem}
            \State $\hat{s} \leftarrow \min\left\{s_{t}, \hat{s}\right\}$
            \State $\bar{\vect{y}}^{[t]} \leftarrow \frac{1}{2} \bar{\vect{y}}^{[t]} + \frac{1}{2}\vect{y}^{[t]}$ \Comment{core point}
            \State $\vectsym{\alpha}^{[t]}, \vectsym{\beta}^{[t]} \leftarrow$ solution of \autoref{eq:additional_problem} with $\bar{\vect{y}}^{[t]}$
            \IIf{$t > \varphi$}
                $\Lambda^\prime_p \leftarrow \Lambda^\prime_p \setminus \{\vectsym{\alpha}^{[t-\varphi]}, \vectsym{\beta}^{[t-\varphi]}\}$
            \EndIIf
            \State $\Lambda^\prime_p \leftarrow \Lambda^\prime_p \cup \{\vectsym{\alpha}^{[t]}, \vectsym{\beta}^{[t]}\}$
            \State $n_{a_t}$ computed according to Corollary~\ref{corr:a_x}
            \State $\matr{Q}^{[\tau]}, \vect{q}^{[\tau]}, \kappa^{[\tau]}$ constructed according to \autoref{eq:qubo}
            \State $\vect{p}^{[t]}, \vect{y}^{[t]} \leftarrow$ solution of \autoref{eq:qubo}
            \If{$\hat{s} - \vect{w}_p\vect{p}^{[t]}\leq \xi$ \textbf{or} $\vect{w}_p\vect{p}^{[t]} > \hat{s}$}
                \If {\autoref{eq:reformulated} with $\vect{y}^{[t]}$ is feasible} 
                \State $m_{t} \leftarrow \argmininline$ of \autoref{eq:reformulated} \Comment{correct solution}
                \State \textbf{Break}
                \EndIf
            \EndIf
        \EndFor
        \IIf{$m_t < 0$} \textbf{return} {\normalfont unknown} \EndIIf
        \State reset $\vect{g}_{n_z-j} \leftarrow 0$
    \EndFor
    \State \textbf{return} {\normalfont robust}
\end{algorithmic}
\end{algorithm}

\begin{corollary}\label{corr:a_x}
Let $\bar{\eta}\in \mathbb{R}$ be the maximum of $\vect{w}_p\vect{p}$ and $e_k = \vectsym{\alpha}^{[k]}\vect{b} + \vectsym{\beta}^{[k]}\vect{d}$, then the minimum number of qubits necessary to approximate $\vect{w}^{[k]}_a\vect{a}^{[k]}$ by a factor $\omega_a$ is
\[n_{a_k} \geq \left\lceil \log_2 \left( \frac{\lvert e_k \rvert + \bar{\eta} + \norm{\vectsym{\alpha}^{[k]}\matr{B}}_1}{\omega_a} + 1\right) \right\rceil .\]
\end{corollary}
\begin{proof}
Let us consider the binary approximation of the cut $\vectsym{\alpha}^{[k]} \left(\vect{b}-\matr{B}\vect{y}\right)+\vectsym{\beta}^{[k]}\vect{d} - \eta + a_k$ previously added to the \textit{master} objective in \autoref{eq:bridge}, i.e. $\vectsym{\alpha}^{[k]} \left(\vect{b}-\matr{B}\vect{y}\right)+\vectsym{\beta}^{[k]}\vect{d} -\vect{w}_p\vect{p} + \vect{w}^{[k]}_a\vect{a}^{[k]}$. 
Starting from the triangle inequality:
\begin{equation*}
\lvert e_k \rvert + \norm{\vect{w}_p\vect{p}}_2 + \norm{\vectsym{\alpha}^{[k]}\matr{B}\vect{y}}_2 \geq \norm{e_k -\vect{w}_p\vect{p} - \vectsym{\alpha}^{[k]}\matr{B}\vect{y}}_2,
\end{equation*}
where $e_k = \vectsym{\alpha}^{[k]}\vect{b} + \vectsym{\beta}^{[k]}\vect{d}$ and $\norm{e_k}_2$ is equal to its absolute value $\lvert e_k \rvert$. As the maximum of $\vect{y}$ is $\mathds{1}_{n_y}$ and the maximum of $\vect{w}_p\vect{p}$ is $\bar{\eta}$, we end at:
\begin{equation*}
\lvert e_k \rvert + \bar{\eta} + \norm{\vectsym{\alpha}^{[k]}\matr{B}}_1 \geq  \lvert e_k \rvert + \norm{\vect{w}_p\vect{p}}_2 + \norm{\vectsym{\alpha}^{[k]}\matr{B}\vect{y}}_2,
\end{equation*}
where $\vectsym{\alpha}^{[k]}\matr{B}\mathds{1}_{n_y}$ is a scalar and so, we can rewrite $\norm{\vectsym{\alpha}^{[k]}\matr{B}\mathds{1}_{n_y}}_2$ as $\norm{\vectsym{\alpha}^{[k]}\matr{B}}_1$ with $\norm{\cdot}_1$ as $l_1$-norm. Let us denote $\bar{a}_k \coloneqq w_{a_k} \cdot \left( 2^{n_{a_k}} - 1\right)$ as the maximum of $\vect{w}^{[k]}_a\vect{a}^{[k]}$, then by imposing
$\vect{w}^{[k]}_a\vect{a}^{[k]} \geq \lvert e_k \rvert + \bar{\eta} + \norm{\vectsym{\alpha}^{[k]}\matr{B}}_1$ we obtain
\begin{equation*}
\bar{a}_k \geq \lvert e_k \rvert + \bar{p} + \norm{\vectsym{\alpha}^{[k]}\matr{B}}_1.
\end{equation*}

\end{proof}

To convert \autoref{eq:qubo} into an Ising model, given that each binary variable satisfies the relation $\vect{x}_j = \vect{x}_j^2$, the linear term $\vect{q}^{[\tau]}$ can be integrated into the quadratic one by adding its entries into the main diagonal.
Finally, we can notice that $\matr{Q}^{[\tau]}$ is symmetric by construction, therefore we can directly replace $\vect{x}$ with $\vect{s} = 2\vect{x}-1$ to obtain the corresponding Ising model. Given these adjustments, \autoref{eq:qubo} can be either solved with a quantum annealing algorithm or QAOA.

\textbf{Final procedure.}
\autoref{alg:hqcran} summarizes the steps described above. Note that $T$ is the upper bound on the number of iterations and $\varphi$ is the maximum number of cuts. In the worst case, if the gap $\xi$ is not met, the algorithm will run for a maximum number of $T$ iterations. 

%% file: sections/results.tex
\section{Results}\label{sec:results}

\textbf{Architecture, Dataset and Training Methods.} We conduct experiments with one type of multilayer perceptron (MLP) neural network: MLP-$2\text{x}[20]$. Here MLP-$m\text{x}[n]$ refers to $m$ hidden layers and $n$ units per hidden layer. The network uses ReLU functions after every fully connected layer. We train our model on the MNIST dataset for $20$ epochs with a batch size of $128$ in two ways: (i) with a regular loss function and (ii) adversarially trained via the Projected Gradient Descent (PGD) as in \citet{madry2017towards}. In case of a regularly trained model, we keep the name MLP-$2\text{x}[20]$, while we refer to it as PGD-$2\text{x}[20]$ if it is adversarially trained. The clean accuracy for the different networks are $95,62\%$ and $86,73\%$ for MLP-$2\text{x}[20]$ and PGD-$2\text{x}[20]$, respectively. The adversarial training used adversarial examples from the infinity norm ball around the input with radius $\epsilon = 0.01$. 

\textbf{Experimental Setup.} We divide the hardware into classical and quantum computing. On the classical side, we ran the algorithm on a server having 4xCPUs Intel(R) Xeon(R) E7-8867 v4 running at 2.40GHz for a total of 72/144 cores/threads. On the quantum side, we used two different system types: quantum annealing and gate-based quantum computers. In the first case, we access the D-Wave Advantage\textsuperscript{TM} system 5.1 constructed with 5760 qubits in two ways: (i) directly on the Quantum Processing Unit (QPU) and (ii) with the Hybrid solver provided by D-Wave Leap\footnote{D-Wave Leap\textsuperscript{TM} Quantum Cloud Service https://cloud.dwavesys.com/leap/}, which makes use of a Tabu search to further decompose the problem and run a part on the QPU. In the second case, we used the QAOA\footnote{Python library Qiskit v0.36.0 https://github.com/Qiskit/qiskit} runtime program of the IBMQ\footnote{IBM Quantum https://quantum-computing.ibm.com/} cloud on the IBM Brooklyn QC having the Hummingbird r2 architecture of 65 qubits, a quantum volume of 32, and 1.5K circuit layers operations per second.

\subsection{Evaluation of Robustness}

We compare our method against a complete (\textbf{Exact})~\cite{tjeng2018evaluating} and incomplete (\textbf{Convex})~\cite{wong2018provable} verifier with the MIP formulation as in \autoref{eq:general_verification} and the convex relaxation from \autoref{eq:convex_relaxation}.
We evaluate the empirical performance of HQ-CRAN in an ideal setting, i.e. where the \textit{master} and \textit{sub} problems are solved using the IBM ILOG CPLEX\footnote{IBM (2020) IBM ILOG CPLEX 20.1 User’s Manual (IBM ILOG CPLEX Division, Incline Village, NV)} software on a classical computer.
For a fair comparison, we propagate the boundaries through the network with interval arithmetic for all methods, and evaluate all methods on the first 100 samples from the MNIST test set with $\epsilon\in\{\frac{1}{255},\frac{2}{255},\frac{4}{255},\frac{8}{255}\}$ as adversarial budget. We use HQ-CRAN with the improvements of \autoref{sec:improvements} and real variables for $\eta$ and $a_k$. 
\begin{table}[ht!]
\captionsetup{font=small}
\caption{
Adversarial robustness of MNIST classifiers to perturbations of $\epsilon$ in a $l_\infty$-norm. 
We run each algorithm on the first 100 test set samples from the MNIST dataset. 
The times are expressed in seconds.
}
\label{table:robustness}
\vskip -0.05in
\begin{footnotesize}
\begin{sc}
\resizebox{\columnwidth}{!}{%

\begin{tabular}{lccccccc}
\toprule
Network &$\epsilon$ &\multicolumn{3}{c}{Certified Accuracy $\uparrow$} &\multicolumn{3}{c}{Average Time [s]$\downarrow$} \\
                    &                    & Exact &HQ-CRAN &Convex &Exact &HQ-CRAN &Convex \\
\midrule
PGD-$2\text{x}[20]$& $\nicefrac{1}{255}$  &$88\%$ &$88\%$ &$88\%$ &$0.04$ &$0.31$ &$0.004$ \\
                    & $\nicefrac{2}{255}$ &$88\%$ &$88\%$ &$88\%$ &$0.06$ &$0.36$ &$0.004$ \\
                    & $\nicefrac{4}{255}$ &$88\%$ &$85\%$ &$85\%$ &$0.10$ &$0.61$ &$0.005$ \\
                    & $\nicefrac{8}{255}$ &$81\%$ &$81\%$ &$61\%$ &$0.13$ &$1.10$ &$0.006$ \\
\midrule
MLP-$2\text{x}[20]$& $\nicefrac{1}{255}$   &$96\%$ &$96\%$   &$96\%$ &$0.05$ &$0.42$ &$0.005$ \\
                    & $\nicefrac{2}{255}$  &$96\%$ &$95\%$  &$95\%$ &$0.07$ &$0.62$ &$0.005$ \\
                    & $\nicefrac{4}{255}$  &$95\%$ &$95\%$  &$95\%$ &$0.10$ &$1.04$ &$0.005$  \\
                    & $\nicefrac{8}{255}$  &$80\%$ &$73\%$  &$24\%$ &$0.14$ &$9.39$ &$0.006$ \\
\bottomrule
\end{tabular}}
\end{sc}
\end{footnotesize}
% \vskip -0.1in
\end{table}

In terms of final robustness certification, we report in Table~\ref{table:robustness} the certified accuracy, i.e., the fraction of verified and correctly classified samples of all test samples. 
To limit the compute time, we set the maximum number of iterations $T$ to 500 and the target gap $\xi$ between the \textit{master} and \textit{sub} problem to $1$. 
Within these settings, HQ-CRAN behaves similarly to \textbf{Convex} up to an $\epsilon$ of $\nicefrac{4}{255}$, and shows better results for $\nicefrac{8}{255}$. 
As previously presented in~\autoref{sec:method}, our method places itself in the middle of \textit{complete} and \textit{incomplete} verification methods. 
The outer approximation of \textbf{Convex} demonstrates poor performances for $\epsilon = \nicefrac{8}{255}$, indicating that outer approximations struggle to verify larger perturbations.
\tom{use beta-crown. if we dont manage to do it here. then we should at least say that there are more advanced approximation techniques}
On the contrary, our average runtime is $\sim2$ orders of magnitude slower than \textbf{Convex}. As this work is only a proof-of-concept and it is expected that the algorithm is slower than both the \textbf{Exact} and \textbf{Convex}, we do not consider current state-of-the-art verifiers based on GPUs, such as $\beta$-CROWN \cite{wang2021betacrown} or GPUPoly \cite{mueller2021scaling}
\tom{+ we should probably include the timings + I dont know the term augmented perturbations}

\subsection{Simulations in Classical Computing}
\begin{figure}[htb!]
\vskip -0.2in
\begin{center}
\includegraphics[width=\columnwidth]{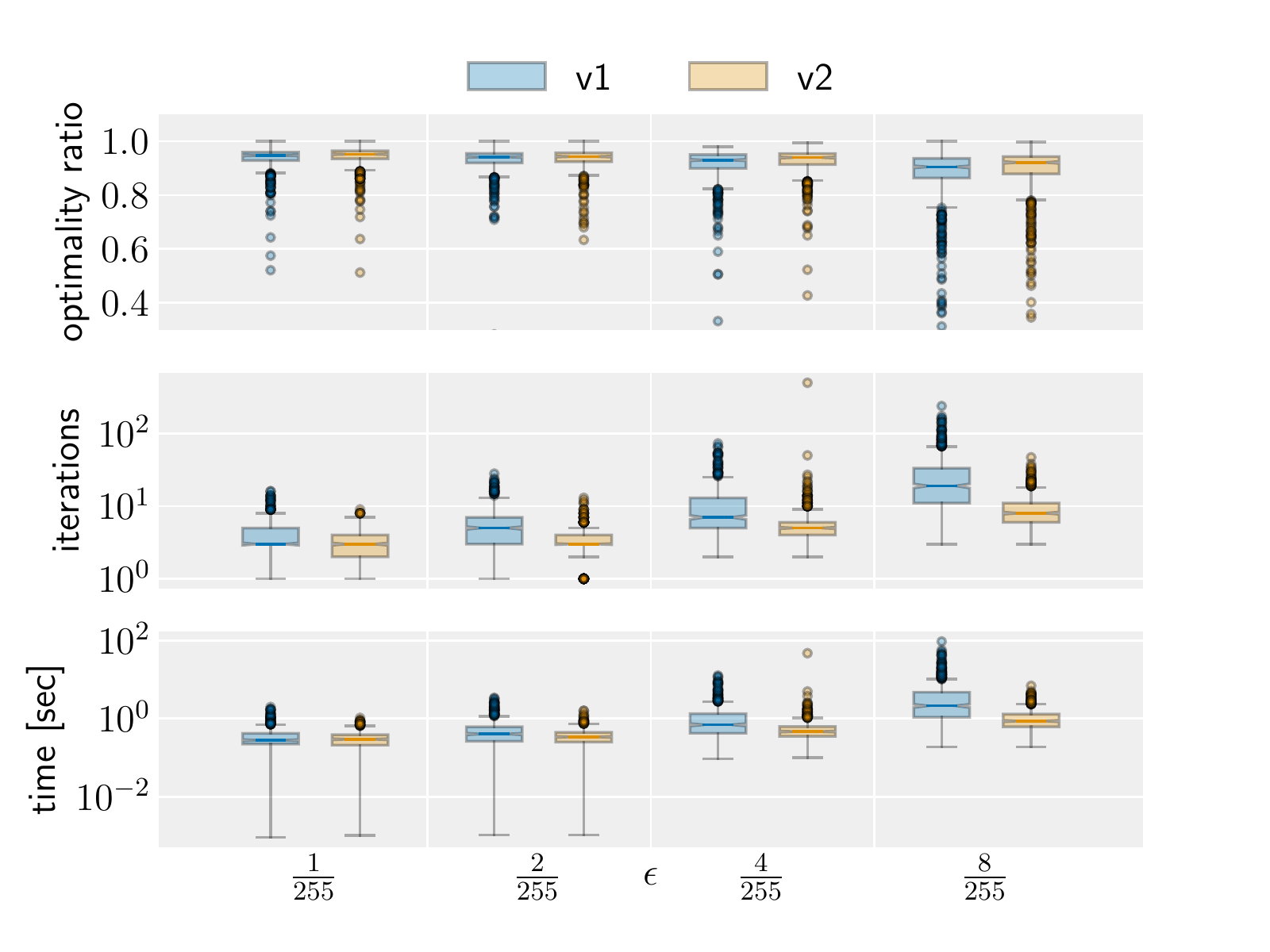}
\captionsetup{font=small}
\caption{
    Optimality ratio between the certified lower bounds of HQ-CRAN, without~(\textbf{v1}) and with~(\textbf{v2}) improvements, and the exact values of \textbf{Exact}. 
    In the context of HQ-CRAN, the bounds for the sub problem and the maximum number of iterations $T$ have been set to $500$. 
    The gap $\xi$ has been set to $1$. 
    We report the distribution as interquartile range of the first 100 samples of the MNIST test set on PGD-$2\text{x}[20]$.
    }
\label{fig:real_pgd}
\end{center}
\vskip -0.3in
\end{figure}
We further conduct evaluations on our improvements presented in \autoref{sec:improvements}. 
For this reason, we structure it as an ablation study and to distinguish the approaches, we define two versions of HQ-CRAN. 
The first version (\textbf{v1}) can be seen as a baseline as it is similar to the earlier proposed algorithm of \cite{chang2020quantum}, which lacks the \textit{additional} problem of \autoref{eq:additional_problem} and the Hamming distance. Instead the second version (\textbf{v2}) uses both. 
All versions compute the minimum number of qubits with Corollaries~\ref{cor:p} and~\ref{corr:a_x}.
We start the experiments with a comparison of the final objective values for each inner minimization of \autoref{eq:general_verification}. 
For a fair evaluation, we compare our results in terms of \textit{optimality ratio} between the final master objective $m_t$ and the solution $m^*$ of \textbf{Exact}, defined as:
\begin{equation*}
    R = \left\{ 
    \begin{aligned}
        &\nicefrac{m^*}{m_t}                     &&\text{if} \; m_t > 0, \\
        &\nicefrac{m_t}{m^*}                     &&\text{else if} \; m^* < 0, \\
        &\nicefrac{m^*}{\left(m^* - m_t\right)}  &&\text{elsewhere.}
    \end{aligned}
    \right.
\end{equation*}
In \autoref{fig:real_pgd}, we show the comparison in terms of optimality ratio, average number of iterations and average time per class. In the context of \textbf{v2}, the results demonstrate a clear reduction on the average number of iterations and time while maintaining a high optimality ratio. These advantages are even more significant for larger $\epsilon$ values. Overall, in both methods the average optimality ratio remains very close to 1. The maximum size of the cuts set $\varphi$ has been set to infinity, i.e. all cuts are considered.

\tom{how come we have consistently worse correct solutions for quantum v2?}
\begin{table*}[htbp]
\captionsetup{font=small}
\caption{Comparison between simulated and quantum annealing on the \textit{master} of HQCRAN, without \textbf{v1} and with \textbf{v2} improvements. The maximum number of iterations $T$ has been set to $15$ and the gap $\xi$ to $1$ and the \textit{sub} problem boundaries $\overline{\vectsym{\alpha}}$ and $\overline{\vectsym{\beta}}$ to $5$. The penalty weights $w_a$ and $w_p$ were set to 0.1 and 0.01, respectively. The maximum size of the cuts set $\varphi$ has been set to infinity and $5$ for \textbf{v1} and \textbf{v2}, respectively. We run each algorithm on the first 100 samples of the MNIST test set.}
\begin{minipage}{\linewidth}
\begin{center}
\begin{footnotesize}
\begin{sc}
\begin{tabular}{cccccccccccccc}
\toprule
\multirow{3}{*}{Networks} & \multirow{3}{*}{$\epsilon$} &\multicolumn{4}{c}{Correct Solutions $\uparrow$} &\multicolumn{4}{c}{Average \# Iterations $\downarrow$} &\multicolumn{4}{c}{Average \# of Qubits $\downarrow$} \\
&  &\multicolumn{2}{c}{Simulated} &\multicolumn{2}{c}{Quantum} &\multicolumn{2}{c}{Simulated} &\multicolumn{2}{c}{Quantum} &\multicolumn{2}{c}{Simulated} &\multicolumn{2}{c}{Quantum} \\
& &\multicolumn{1}{c}{v1} &\multicolumn{1}{c}{v2} &\multicolumn{1}{c}{v1} &\multicolumn{1}{c}{v2} &\multicolumn{1}{c}{v1} &\multicolumn{1}{c}{v2} &\multicolumn{1}{c}{v1} &\multicolumn{1}{c}{v2} &\multicolumn{1}{c}{v1} &\multicolumn{1}{c}{v2} &\multicolumn{1}{c}{v1} &\multicolumn{1}{c}{v2} \\
\midrule
\multirow{4}{*}{PGD-$2$x$[20]$}
&$\nicefrac{1}{255}$ &$94\%$    &$96\%$  &$74\%$  &$73\%$ &$4\pm2$ &$3\pm2$     &$4\pm3$ &$3\pm2$   &$59\pm35$  &$45\pm21$  &$52\pm45$ &$37\pm17$ \\
&$\nicefrac{2}{255}$ &$82\%$    &$91\%$  &$54\%$  &$50\%$ &$6\pm3$ &$4\pm3$     &$4\pm3$ &$3\pm2$   &$82\pm43$  &$58\pm23$  &$65\pm48$ &$46\pm20$ \\
&$\nicefrac{4}{255}$ &$52\%$    &$67\%$  &$36\%$  &$34\%$ &$8\pm3$ &$7\pm4$    &$5\pm4$  &$5\pm3$   &$121\pm50$ &$77\pm20$  &$80\pm54$ &$59\pm21$ \\
&$\nicefrac{8}{255}$ &$10\%$    &$17\%$    &$13\%$   &$11\%$  &$12\pm3$ &$11\pm4$   &$9\pm5$  &$7\pm4$   &$186\pm44$ &$94\pm10$   &$139\pm70$   &$81\pm19$ \\
\midrule
\multirow{4}{*}{MLP-$2$x$[20]$} 
&$\nicefrac{1}{255}$ &$83\%$    &$85\%$ &$65\%$ &$61\%$     &$6\pm3$    &$5\pm3$    &$5\pm4$ &$4\pm2$ &$88\pm47$  &$64\pm22$ &$77\pm53$ &$52\pm20$ \\
&$\nicefrac{2}{255}$ &$61\%$    &$65\%$ &$38\%$ &$37\%$     &$8\pm3$    &$7\pm4$    &$5\pm4$ &$5\pm4$ &$112\pm51$ &$75\pm20$ &$80\pm54$ &$60\pm21$  \\
&$\nicefrac{4}{255}$ &$27\%$    &$30\%$ &$17\%$ &$17\%$     &$10\pm4$   &$10\pm4$    &$6\pm4$ &$7\pm5$ &$155\pm56$ &$88\pm15$ &$93\pm61$ &$72\pm24$ \\
&$\nicefrac{8}{255}$ &$3\%$     &$4\%$  &$4\%$ &$6\%$     &$14\pm2$   &$14\pm2$   &$12\pm5$ &$10\pm5$ &$214\pm30$ &$100\pm5$ &$179\pm66$ &$93\pm16$ \\
\bottomrule
\end{tabular}
\end{sc}
\end{footnotesize}
\end{center}
\end{minipage}
\label{tab:annealing}
\end{table*}

\subsection{Quantum Hardware Experiments}

We continue our evaluation with a comparison between simulated\footnote{Python dwave-neal v0.5.7, general Ising graph simulated annealing solver.} and quantum annealing. 
Embedding the problem into a QA requires solving the minor embedding problem. 
In our case, we use the clique embedding strategy \cite{boothby2016fast} which handle the connectivity problem in a way that each connection larger than the maximum available size (in case of Advantage it is $\sim$15) is handled with long chains between qubits. 
To decrease the chains, before submitting the model to the sampler, we prune connection below a certain threshold (5$\%$) and retaining only interactions with values commensurate with the sampler’s precision. In both simulated and quantum annealing, the number of reads was set to 100 and in the context of simulated annealing, the number of sweeps was set to 50000.

In \autoref{tab:annealing}, we report the numerical results in terms of percentage of correct solutions, i.e. solutions where $m_t \leq m^*$. 
The average number of qubits reflects the size of the \textit{master} problem before being embedded into the quantum hardware. 
We notice that the number of correct solutions decrease for increased values of $\epsilon$. The behaviour has to be attribute to the maximum number of iterations $T$ fixed to $15$. 
This choice limits the convergence of the algorithm but gives us a fair comparison between the approaches. 
As previously mentioned, the limiting factor is the clique embedding which fixes the maximum size of logical qubits to $177$ on the Advantage~\cite{mcgeoch2021advantage}. 
The \textbf{v2} shows an overall reduced number of qubits with respect to \textbf{v1} while maintaining approximately the same correctness. 
\tom{I would extract the settings here. It is a bit hard to read still. Go maybe like this: in the quantum experiment we use chain ... logical qubits ... max iterations ...  and then: In table 2 we can see the results and we see that v2 performs... I think this will make it better and is basically how you did it before}\\

Finally, we evaluate HQ-CRAN running on a gate-based quantum computer via QAOA and on a quantum annealer via Leap.
Due to time limits on the quantum hardware, we present results on the first and the first ten MNIST test set samples for QAOA and Leap, respectively.
This also restricts our hyperpameter choices for QAOA to the Perturbation Stochastic Approximation optimizer with a maximum iteration number of 10 instead of 100, 2 repetitions and 1024 shots. 
For Leap we use the default 3-minute execution time limit.
The penalty weights $w_a$ and $w_p$ were set to 0.1 and 0.01, respectively. $\epsilon=\frac{1}{255}$ and a target gap $\xi$ of 1 for both algorithms.
With this setting, QAOA starts of with $\sim28$ qubits and requires 13 additional qubits per step.
For Leap, the number of theoretical qubits is identical but one needs to add qubits required for increased connectivity as discussed previously.

\begin{figure}[t]
\captionsetup{font=small}
%\vskip 0.2in
\begin{center}
\centerline{\includegraphics[width=\columnwidth]{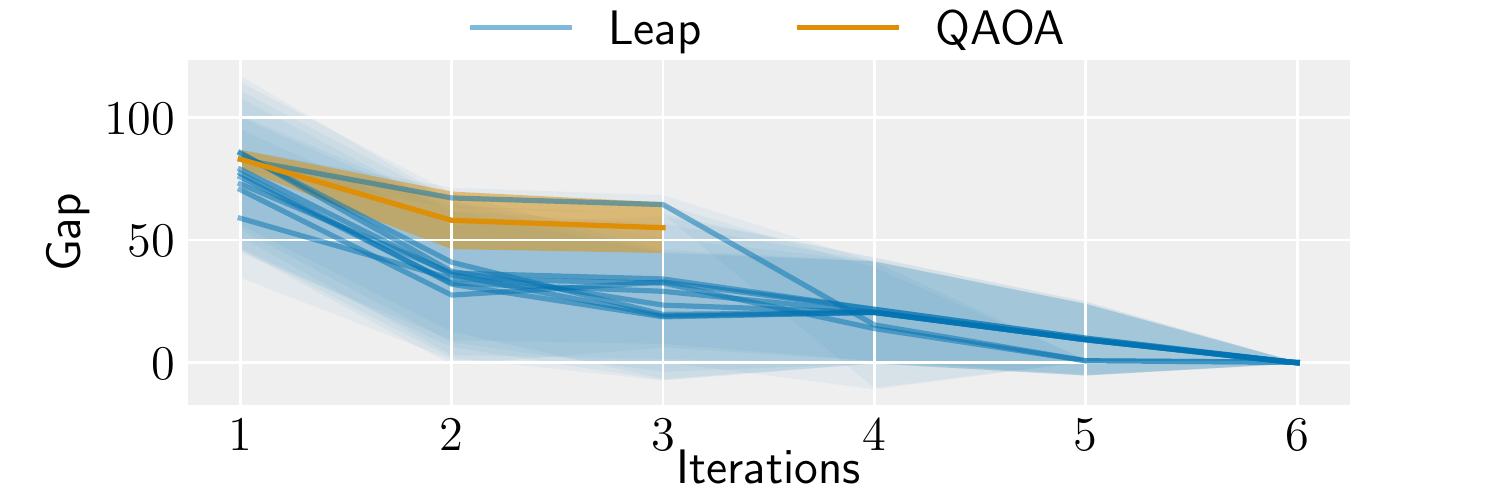}}
\caption{Average and standard deviation of the difference between \textit{master} and \textit{sub} objective at each iteration for QAOA and Leap. In the case of Leap, we plot the first 10 samples of the MNIST test set. In contrast, for QAOA, we only show the first one. We used an $\epsilon$ value of $\nicefrac{1}{255}$ and PGD-$2\text{x}[20]$.}
\label{fig:hardware_steps}
\end{center}
\vskip -0.3in
\end{figure}

\autoref{fig:hardware_steps} shows the convergence behavior of Leap and QAOA.
The gap is defined as difference between the objectives of \textit{master} and \textit{sub} problems at each iteration. The solution is reached when the gap is lower than our target gap $\xi=1$.
While we find QAOA not converging by the constrained accessibility and limitations of current QC hardware, Leap converges close to the optimal solution within 6 steps for all tested samples.
\tom{i did not see where the GAP is defined. maybe i missed it from another experiment. but I would make it clear here to the reader again that this is the gap between master and sub and if the gap is zero and the point is valid we achieved the exact solution. and maybe also specify the remaining gap?}

%% file: sections/conclusion.tex
\section{Discussion}

As discussed in \autoref{sec:results}, the verification abilities of our new HQ-CRAN method are promising with exact solvers such as CPLEX. In addition, the proposed improvements achieved better results in terms of number of iterations and qubits than other hybrid Benders decomposition methods. Apart from these advantages, the results on quantum devices are still far from optimal in the case of larger neural networks or higher perturbations. There are multiple reasons for this which we want to discuss here.
Firstly, due to decoherence of qubits or thermal fluctuations, the discrepancy between quantum simulation and NISQ devices is a known issue impeding results and the ability to detect speed-ups~\cite{ronnowDefiningDetectingQuantum2014,katzgraberGlassyChimerasCould2014}.
Secondly, the approximation of the real-valued variables is sub-optimal as it is sensitive towards random bit flips as some bits have exponentially larger impact than others.
Further, the approximation of real-valued numbers may change the optimal value of our optimization.
Thirdly, we set $\overline{\vectsym{\alpha}}$ and $\overline{\vectsym{\beta}}$ to arbitrary large enough values, even though this may lead to many unnecessary cuts increasing time, qubit requirements and introducing numerical instabilities.

Despite these limitations, future improvements to quantum hardware and algorithms may alleviate these issues.
The first issue may be eliminated by time as quantum computers continue to grow exponentially in the number and qubits~\cite{henrietQuantumComputingNeutral2020}.
Moreover, improvements to qubit connectivity and error-correction may close the gap between quantum simulations and hardware, guaranteeing the soundness of HQ-CRAN.
Additionally, instead of binary approximations for real-valued variables, angular or phase encodings may be more efficient and less error prone~\cite{wiebeFloatingPointRepresentations2013,haenerQuantumCircuitsFloatingpoint2018}.
Future research may also lead to bounds for $\overline{\vectsym{\alpha}}$ and $\overline{\vectsym{\beta}}$ to reduce the number of iterations without interfering with the optimal solution.
Lastly, introducing new stopping criteria may further reduce the number of iterations, e.g., due to the nature of robustness verification, one could stop after the master objective gets larger than zero or the sub objective falls below zero as this is sufficient to issue or deny a certificate.

In the end, one may obtain an polynomial speed-up~\cite{groverFastQuantumMechanical1996}, though, still in exponential time as it is assumed that QC may not solve NP-hard problems with high probability in P time~\cite{mohrQuantumComputingComplexity2014}. 
An analogy of this can be observed in Benders decomposition, as in the worst case, one has to add all possible cuts from the exponentially large set of extreme rays.
%\tom{not sure if that sentence is correct. basically each iteration of benders decomposition solves already one NP hard problem in the master. or i dont really get what you mean here}
Luckily, in practice, few iterations already lead to sufficient results. 

\section{Conclusion}
We proposed a proof of concept for a hybrid quantum algorithm for robustness verification leveraging the computational power of NISQ devices.
We have shown that by applying Benders decomposition to the MIP related to ReLU network certification one obtains a classically solved LP and a QC suitable sub problem.
Further, we proved the soundness of our certificate and provide estimates on the number of qubits to achieve reasonable certificate quality.
In our experiments, we observed that directly implementing Benders algorithm often leads to slow convergence, requiring many iterations and more qubits than current quantum hardware offers. 
Nevertheless, we improved related works by reducing the number of iterations and placing a limit to the overall number of qubits required.
In conclusion, this work paves the way towards exploiting QC combinatorial power to accelerate robustness verification of ReLU networks in the future.